\documentclass[3p]{elsarticle}
\usepackage[utf8]{inputenc}
\usepackage[british]{babel}
\usepackage{numcompress}
\usepackage{amssymb}
\usepackage{amsthm}
\usepackage{amsmath}
\usepackage{amsfonts}
\usepackage{mathpartir}
\usepackage{enumerate}
\usepackage{hyperref}
\usepackage{hyphenat}
\usepackage[shortcuts]{extdash}
\usepackage{color}
\usepackage{tikz}
\usetikzlibrary{matrix,arrows,positioning}

\hypersetup{%
  pagebackref,%
  pdftex,%
  final=true,%
  colorlinks=true,%
  linkcolor=black,%
  citecolor=black,%
  urlcolor=black}

\tikzset{bubble/.style={rectangle,minimum size=6mm,rounded corners=3mm,
                        draw=black},
         plain/.style={rectangle,minimum size=6mm,rounded corners=3mm},
         bar/.style={ultra thick,draw=gray!30},
         trans/.style={ultra thick,draw=gray!30,->},
         flow/.style={dashed,->}}

\newcommand{\ie}{i.e.}
\newcommand{\eg}{e.g.}
\newcommand{\ow}{otherwise}

\newcommand{\oTerm}{\mathbb{T}(\Sigma)}
\newcommand{\cTerm}{\mathbb{C}(\Sigma)}
\newcommand{\pset}{\mathcal{P}}
\newcommand{\Trans}{\longrightarrow}
\newcommand{\Label}{L}
\newcommand{\oSource}[2][\id]{\mathbb{S}_{#2}^{#1}(\Sigma)}
\newcommand{\oTarget}[2][\id]{\mathbb{R}_{#2}^{#1}(\Sigma)}
\newcommand{\cSource}[2][\id]{\mathbb{O}_{#2}^{#1}(\Sigma)}
\newcommand{\cTarget}[2][\id]{\mathbb{D}_{#2}^{#1}(\Sigma)}

\newcommand{\NN}{\mathbb{N}}

\newcommand{\rel}[1]{\buildrel #1\over\longrightarrow}
\newcommand{\nrel}[1]{\,\,\,\not\!\!\!\rel{#1}}

\newcommand{\Rel}{\,\mathcal{R}\,}
\newcommand{\defeq}{\buildrel\rm def\over=}

\newcommand{\prj}{\mathit{prj}}

\newcommand{\id}{\mathit{id}}

\newcommand{\var}{\mathrm{var}}

\newcommand{\DT}{\mathcal{D}}
\newcommand{\tr}{\mathit{Tr}}

\theoremstyle{plain}% default
\newtheorem{thm}{Theorem}
\newtheorem{lem}{Lemma}
\newtheorem{prop}{Proposition}
\newtheorem{cor}{Corollary}

\theoremstyle{definition}
\newtheorem{defn}{Definition}
\newtheorem{exmp}{Example}
\newtheorem{obs}{Observation}
\newtheorem{ntn}[defn]{Notation}
\newtheorem*{note}{Note to the reader}

\theoremstyle{remark}
\newtheorem{rem}{Remark}

\newproof{pf}{Proof}

\newcounter{ai}
\newcounter{li}
\newcounter{ani}
\newcounter{ii}

\begin{document}

\begin{frontmatter}
\title{A unified rule format for bounded nondeterminism\\in SOS with terms as
  labels\tnoteref{t1}}
\tnotetext[t1]{This research has been supported by the project `Nominal
  Structural Operational Semantics' (nr.~141558-051) of the Icelandic Research
  Fund, the projects StrongSoft (TIN2012-39391-C04–04) and TRACES
  (TIN2015-67522-C3-3-R) of the Spanish Ministerio de Econom\'{\i}a y
  Competitividad, the project 001-ABEL-CM-2013 within the NILS Science and
  Sustainability Programme, and the grant GR3/14 of the Universidad
  Complutense de Madrid - Banco Santander.}

\author[ru]{L. Aceto\corref{cor1}}
\ead{luca@ru.is}

\author[ru]{I. Fábregas}
\ead{fabregas@ucm.es}

\author[ru]{A. García-Pérez}
\ead{alvarog@ru.is}

\author[ru]{A. Ingólfsdóttir}
\ead{annai@ru.is}

\address[ru]{ICE-TCS, School of Computer Science, Reykjavik University,\\
  Menntavegur 1, IS-101, Reykjavik, Iceland.}

\cortext[cor1]{Corresponding author}

%\titlerunning{Bounded nondeterminism in nominal SOS}

%\authorrunning{Luca Aceto et al.}

\begin{abstract}
  We present a unified rule format for structural operational semantics with
  terms as labels that guarantees that the associated labelled transition
  system has some bounded-nondeterminism property. The properties we consider
  include finite branching, initials finiteness and image finiteness.

  % We introduce a unified framework for rule formats that ensure bounded
  % nondeterminism in structural operational semantics. Our framework allows to
  % have transitions that are labelled by arbitrary terms, which, together with
  % a cofiniteness result for the freshness assertions of the nominal setting,
  % enables to apply the rule formats to nominal structural operational
  % semantics.
\end{abstract}


\end{frontmatter}

%%%%%%%%%%%%%%%%%%%%%%%%%%%%%%%%%%%%%%%%%%%%%%%%%%%%%%%%%%%%%%%%%%%%%%%%%%%
%% Intro
\section{Introduction}
\label{sec:introduction}

Structural operational semantics (SOS) \cite{Plo81,Plo04} is a widely used
formalism for defining the formal semantics of computer programs and for
proving properties of the corresponding programming languages. In the SOS
formalism a transition system specification (TSS) \cite{GV92}, which consists
of a signature together with a set of inference rules, specifies a labelled
transition system (LTS) \cite{Kel76} whose states (\ie, processes) are closed
terms over the signature and whose transitions are those that can be proved
using the inference rules.

Rule formats \cite{AFV01,MRG07} are (syntactically checkable) restrictions on
the inference rules of a TSS that guarantee some useful property of the
associated LTS. In this paper we focus on the finiteness of the number of
outgoing transition from a given state, which is referred to as bounded
nondeterminism in \cite{Gla87}. Broadly, bounded nondeterminism is taken as a
synonym of finite branching \cite{FV03}. Finite branching breaks down into the
more elementary properties of initials finiteness and image finiteness
\cite{Abr87,HM85}. Intuitively, initials finiteness focuses on the finiteness
of the number of initial actions a process can perform, whereas image
finiteness focuses on the finiteness of the number of states a process can
reach by performing a given action.

The wealth of research on various kinds of bounded nondeterminism properties
has been motivated by results witnessing desirable features enjoyed by the
LTSs that afford them. For example, bisimulation-based preorders over finitely
branching LTSs with divergence are ``finitary'', and this property makes it
possible to use them as behavioural yardsticks to prove full abstraction
results for denotational semantics for process calculi---see, for
instance,~\cite{Abramsky91,AcetoI97}. This state of affairs is in stark
contrast with the general argument presented in~\cite{AP86} to the effect that
no continuous least fixed-point semantics, satisfying a certain full
abstraction property, can exist in the presence of countable
nondeterminism. Another classic example of the use of a bounded nondeterminism
property from the literature on concurrency theory and modal logic is provided
by the Hennessy-Milner Theorem, to the effect that two processes in an
image-finite LTS are bisimilar exactly when they satisfy the same formulae in
Hennessy-Milner logic~\cite{MR1837791,HM85}. Languages whose semantics is
given by LTSs satisfying some bounded nondeterminism property also tend to
have lower ``complexity'' than those exhibiting unbounded nondeterminism. For
instance, Chandra has shown in~\cite{Chandra78} that the halting problem for
programs with unbounded nondeterminism is complete for $\Pi_1^1$ and is
therefore of higher complexity than truth in the standard model of the natural
numbers. Indeed, in the presence of infinite nondeterminism, quoting
from~\cite[page~131]{Chandra78}:
\begin{quote}
total correctness (or even absence of divergence) is not expressible in
first-order arithmetic. This may be contrasted with the results
in~\cite{HarelMP77} where the model does not allow infinite nondeterminism.
\end{quote}
In the light of the above-mentioned work, it should come as no surprise that
the study of rule formats for bounded nondeterminism is far from being
new. Finite branching of the associated LTS is one of the sanity properties
guaranteed by the GSOS format of \cite{BIM95}. Vaandrager \cite{Vaa93}
introduced a rule format for SOS, based on the de~Simone format \cite{deS85},
that guarantees that the associated LTS is finite branching. Following
Vaandrager, Bloom \cite{Blo94} introduced a rule format for his CHOCOLATE
formalism that also guarantees a finite-branching LTS. Finally, Fokkink and Vu
\cite{FV03} introduced yet another, less restrictive rule format for SOS, which
relies on an adaptation of the notion of strict stratification from
\cite{Gro93}, and showed that a TSS in this format induces an LTS that is
finite branching. Recently, we adapted the work in \cite{FV03} and presented
rule formats also for initials finiteness and image finiteness \cite{AcetoGI16}.

Our work in this paper takes this programme further by advancing on three fronts:
\begin{enumerate}[(i)]
\item We extend the applicability of the rule formats for bounded nondeterminism to the case of
  higher-order processes by tackling one of the most prominent challenges in
  that setting, namely that of allowing arbitrary terms as labels in
  transitions \cite{Ber98,MGR05}.
\item We consider a family of bounded-nondeterminism properties that are more
  elementary than finite branching, and which include image finiteness and
  initials finiteness \cite{Abr87}. We introduce a unified rule format that
  can be instantiated to each of twelve bounded-nondeterminism conditions.
\item For each rule in a TSS, our rule format uses \emph{global} information
  in order to determine whether the rule is a so-called junk rule (\ie, a rule
  that is never involved in proving some transition). Our rule format filters out more junk
  rules than the one in \cite{FV03}.
\end{enumerate}
Previous works on rule formats for bounded nondeterminism, such as that in
\cite{FV03}, require the labels in transitions to be ground. However, SOS for
higher-order processes \cite{Ber98,MGR05} needs to consider arbitrary terms as
labels and the use of variables in labels must be suitably restricted to
guarantee the finiteness property of interest. Indeed, the unconstrained
occurrence of variables in a label could break various bounded nondeterminism
properties because a variable stands for an arbitrary process of the LTS. As
an example, an axiom $f(x)\rel{y}f(x)$, where $x$ and $y$ are variables, could
be used to prove any transition of the form $f(p)\rel{q}f(p)$, where $p$ and
$q$ are arbitrary closed terms. Therefore any term of the form $f(p)$ would be
neither finite branching nor initials finite. On the other hand, the use of
$y$ in that axiom does not jeopardize image finiteness.

In order to deal with (i) and (ii) above in a uniform way for all the
twelve bounded nondeterminism properties studied in the paper, we consider 
a family of representational transformations for
transitions that, instead of \emph{triadic} transition formulae $t\rel{l}t'$,
consider \emph{dyadic} transition formulae $s\rel{} r$ in which the label $l$
occurs either as a part of the source $s$ or of the target $r$. Building on \cite{FV03},
the dyadic transformations make it possible to define a general rule format (see
Theorem~\ref{thm:dk-finiteness}), which can be instantiated with the
particular dyadic transformation that corresponds to a bounded nondeterminism property of our choice.

The dyadic transformations unveil a space of bounded-nondeterminism properties
that has a rich structure (see Section~\ref{sec:other-properties}). The space
contains some elementary properties that are minimal in the ``information''
order. We explore the equivalences between the conjunction of certain pairs of
properties and other properties that are derived, bringing forth a principled
approach to bounded nondeterminism in all its possible variants. 

In order to deal with the detection of junk rules, point (iii) above, we refine the $\eta$-types of \cite{FV03},
which are at the basis of the rule formats for bounded nondeterminism from
\cite{AFV01,FV03}.
%If a rule gives rise to transitions from a given process,
%the rule would have a valid $\eta$-type. 
Junk rules are those that can never take part in a proof tree. The set of TSSs for which a rule
format is applicable can be safely enlarged by relaxing the restrictions of
the rule format on junk rules, since these rules cannot give rise to
transitions. We introduce the $S$-types (see Definition~\ref{def:S-types}),
which use global information (\ie, information from every rule in a TSS) in
order to detect more junk rules than the $\eta$-types of \cite{FV03}, which
only use local information. A key ingredient of the $S$-types is a notion of
\emph{partial} strict stratification (see Definition~\ref{def:partial-strict-stratification}) 
that allows some processes that unify with sources of
premises in a rule to have undefined stratification order.

% The $\eta$-types determine whether a rule may give rise to
% transitions or not by using information that is local to the rule. The
% so-called junk rules are those that can never take part in a proof tree. The
% set of TSSs for which a rule format is applicable can be safely enlarged by
% relaxing the restrictions of the rule format on junk rules, since these rules
% cannot give rise to transitions. We introduce the $S$-types (see
% Definition~\ref{def:S-types}), which use global information (\ie, information
% from every rule in the TSS) in order to improve the detection of junk rules. A
% key ingredient of the $S$-types is a notion of partial strict stratification
% (see Definition~\ref{def:partial-strict stratification}) that allows some
% processes that unify with sources of premisses in a rule to have undefined
% order.

The rest of the
paper is organized as follows. Section~\ref{sec:preliminaries} collects
standard background on the SOS formalism and unifies
notation. Section~\ref{sec:dyadic-transformation} introduces the dyadic
transformations, which provide a uniform framework to deal with TSSs with
terms and labels and their bounded nondeterminism
properties. Section~\ref{sec:other-properties} studies the structure of the
space of bounded-nondeterminism properties unveiled by the dyadic
transformations. Section~\ref{sec:junk-rules-eta-types} explains how the
set of finitely branching TSSs that meet the conditions of the rule format
from \cite{FV03} can be enlarged by considering the so-called junk
rules. Section~\ref{sec:refining-eta-types} shows how to refine the
$\eta$-types of \cite{FV03} in order to filter out more junk
rules. Section~\ref{sec:s-types} introduces the $S$-types, which consider the
previously mentioned refinements on the $\eta$-types and adapt them to dyadic
TSSs. Section~\ref{sec:rule-formats} presents rule formats that are based on
the $S$-types guaranteeing finite branching of the LTS associated to a triadic
TSS with terms as labels. (A variety of bounded nondeterminism properties can
be guaranteed by picking different dyadic transformations.)
Section~\ref{sec:application} shows examples of the applicability of our rule
format. In particular, we consider a TSS with terms as labels that implements
a subset of CHOCS \cite{MGR05,Tho89} and use our format to show various bounded 
nondeterminism properties of its transition relations. Section~\ref{sec:future-work} discusses
avenues for future work and concludes. Throughout the paper, we use a variety 
of simple (but admittedly artificial) examples to motivate our definitions and main 
observations.

% Section~\ref{sec:junk-rules-eta-types} explains how the detection of junk
% rules is paramount in order to extend a rule
% format. Section~\ref{sec:refining-eta-types} provides motivational examples to
% our refinements on the $\eta$-types, which increase the number of TSSs that
% are covered by the rule format in \cite{FV03}. Section~\ref{sec:s-types}
% introduces the $S$-types, whose key ingredients are a notion of partial strict
% stratification (Definition~\ref{def:partial-strict stratification}) and a
% restricted support map (Definition~\ref{def:restricted-support}). The partial
% strict stratification allows some processes that unify with the premisses of
% rules to have undefined order. The restricted support map considers global
% information in order to determine if the premisses of a rule can unify with
% other rules in the TSS. Section~\ref{sec:rule-formats} presents our rule
% format for bounded nondeterminism, which is based on the $S$-types. The rule
% format, which is parametric on a dyadic transformation, can be instantiated to
% any bounded-nondeterminism property by picking the dyadic transformation that
% corresponds to the property of our choice.

%%%%%%%%%%%%%%%%%%%%%%%%%%%%%%%%%%%%%%%%%%%%%%%%%%%%%%%%%%%%%%%%%%%%%%%%%%%
%% Preliminaries
\section{Preliminaries}
\label{sec:preliminaries}

In this section we review and adapt standard background on the structural
operational semantics formalism (SOS for short) and unify notation. We follow
the presentation in \cite{FV03}.

For a set $S$, we write $\pset(S)$ for the collection of all the subsets of
$S$, and $\pset_\omega(S)$ for the collection of all the finite subsets of
$S$.

\begin{defn}[Signature and Term]
  \label{def:terms}
  We assume a countably infinite set of variables $V$, ranged over by $x$,
  $y$, $z$. A signature $\Sigma$ is a set of function symbols, disjoint from
  $V$, together with an arity map that assigns a natural number to each
  function symbol. We use $f$ and $g$ to range over the set of functions
  symbols in $\Sigma$. Function symbols of arity zero, which are ranged over
  by $c$ and $d$, are called constants. Function symbols of arity one and two
  are called unary and binary functions respectively.

  The set $\oTerm$ of (open) terms over a signature $\Sigma$, ranged over by
  $t$, $u$, $v$ is the least set such that:
  \begin{enumerate}[(i)]
  \item each variable is a term, and
  \item if $f$ is a function symbol of arity $n$ and $t_1, \ldots, t_n$ are
    terms, then $f(t_1,\ldots,t_n)$ is a term.
  \end{enumerate}

  The function $\var : \oTerm\to \pset_\omega(V)$ delivers, for a term $t$,
  the set of variables that occur in $t$. A term $t$ is closed iff
  $\var(t)=\emptyset$. The set of closed terms over $\Sigma$, ranged over by
  $p$, $q$, is denoted by $\cTerm$.
\end{defn}

Intuitively, a signature $\Sigma$ collects the term constructors offered by a
programming or specification language. In this paper we consider terms also as
transition labels, where symbols $l$, $m$, $n$ range over label terms.
\begin{defn}[Formula]
  The set of positive formulae over signature $\Sigma$ is the set of triples
  $(t,l,t')\in\oTerm\times \oTerm\times\oTerm$. We use the more suggestive
  notation $t\rel{l}t'$ in lieu of $(t,l,t')$. The set of negative formulae
  over signature $\Sigma$ is the set of pairs $(t,l)\in\oTerm\times
  \oTerm$. We use the more suggestive notation $t\nrel{l}$ in lieu of $(t,l)$.
\end{defn}

\begin{defn}[Substitution]
  A substitution is a partial map $\sigma:V\to\oTerm$. The substitutions are
  ranged over by $\sigma$, $\tau$. A substitution is closed if it maps
  variables to closed terms. A substitution extends to a map from terms to
  terms in the usual way, \ie, the term $\sigma(t)$ is obtained by replacing
  the occurrences in $t$ of each variable $x$ in the domain of $\sigma$ by
  $\sigma(x)$. When applying substitutions $\sigma$ and $\tau$ successively,
  we may abbreviate $\tau(\sigma(t))$ to $\tau\sigma(t)$. We say term $u$ is a
  substitution instance of $t$ iff there exists a substitution $\sigma$ such
  that $\sigma(t)=u$.

  In what follows, we will sometimes use the notation $\{x_i \mapsto t_i \mid
  i\in I\}$, where $I$ is an index set and the $x_i$'s are pairwise distinct
  variables, to denote the substitution that maps each $x_i$ to the term $t_i$,
  with $i\in I$.
\end{defn}
A substitution $\sigma$ extends to formulae $t\rel{l}t'$ and $u\nrel{l}$ in
the usual way, by applying the substitution to each term component of the
formula, \ie, ${\sigma(t)\rel{\sigma(l)}\sigma(t')}$ and
$\sigma(u)\nrel{\sigma(l)}$ respectively. The notion of substitution instance
extends similarly.

Labelled transition systems~\cite{Kel76} are a fundamental model of
computation and are often used to describe the operational semantics of
programming and specification languages (see, for instance,
\cite{Mil89,Plo83,Plo04,SW01}). We will consider classic labelled transition
systems~\cite{Kel76} where transitions are labelled with processes instead of
actions.
\begin{defn}[Labelled transition system with processes as labels]
  \label{def:lts}
  Let $\Sigma$ be a signature. A labelled transition system with processes as
  labels (LTS for short) is a pair $(\cTerm, \to)$ where $\cTerm$ is the set
  of processes, \ie, closed terms over $\Sigma$, and $\Trans \subseteq
  \cTerm\times \cTerm\times\cTerm$ is the set of transitions, \ie, closed
  positive formulae. We say that $p\rel{l}p'$ is a transition of the LTS iff
  $(p,l,p')\in\Trans$.
\end{defn}

Transition system specifications are sets of inference rules that can be used
to prove the valid transitions between terms in a language.
\begin{defn}[Transition system specification with terms as labels]
  Let $\Sigma$ be a signature. A transition rule (a
  rule, for short) $\rho$ is of the form
  \begin{mathpar}
    \inferrule*[]
    {H}
    {t\rel{l}t'}
  \end{mathpar}
  (abbreviated as $H/t\rel{l}t'$) where $H$ is a set of positive premises of
  the form $u\rel{m}u'$ and negative premises of the form $v\nrel{n}$, and
  $t\rel{l}t'$ is the conclusion of the rule (with
  $t,t',u,u',v,l,m,n\in\oTerm$).  We say $t$ is the source, $l$ is the label,
  and $t'$ is the target of $\rho$. We say $\rho$ is an axiom iff $\rho$ has
  an empty set of premises, \ie, $H=\emptyset$.

  A transition system specification with terms as labels (TSS for short) is a
  set of transition rules.
\end{defn}
A substitution map extends to a rule $\rho$ by applying the substitution to
the formulae in $\rho$. The notion of substitution instance extends similarly
to rules.

\begin{defn}[Unify with a rule]
  Let $R$ be a TSS. We say that transition $p\rel{l}p'$ unifies with rule
  $\rho\in R$ iff $\rho$ has conclusion $t\rel{m}t'$ and $p\rel{l}p'$ is a
  substitution instance of $t\rel{m}t'$.
\end{defn}

\begin{defn}[Proof tree]
  \label{def:proof-tree}
  Let $R$ be a TSS without negative premises. A proof tree in $R$ of a
  transition $p\rel{l}p'$ is an upwardly branching tree without paths of
  infinite length whose nodes are labelled by transitions such that
  \begin{enumerate}[(i)]
  \item the root is labelled by $p\rel{l}p'$, and
  \item if $K$ is the set of labels of the nodes directly above a node with
    label $q\rel{n}q'$, then $K/q\rel{n}q'$ is a substitution instance of some
    rule ${H/t\rel{m}t'}\in R$.
  \end{enumerate}
  We say that $p\rel{l}p'$ is provable in $R$ iff $p\rel{l}p'$ has a proof
  tree in $R$.
\end{defn}
The set of provable transitions in $R$ is the least set of transitions that
satisfies the rules in $R$. Notice that if $p\rel{l}p'$ unifies with an axiom
(\ie, a rule of the form $\emptyset/t\rel{m}t'$) then, trivially, $p\rel{l}p'$
has a proof tree in $R$ which consists of a root node labelled by
$p\rel{l}p'$.

A TSS with terms as labels and without negative premises induces an LTS with
processes as labels in a straightforward way.
\begin{defn}[TSS induces LTS]
  \label{def:tss-induced-lts}
  Let $R$ be a TSS with terms as labels and without negative premises, and let $T$
  be an LTS with processes as labels. We say that $R$ induces $T$ (or $T$ is associated with
  $R$) iff the set of transitions of $T$ is the set of provable transitions in
  $R$.

  The phrases
  \begin{itemize}
  \item $p\rel{l}p'$ is provable in $R$,
  \item $p\rel{l}p'$ is a transition of $T$, and
  \item $p$ can perform an $l$-transition to $p'$ in $T$
  \end{itemize}
  are synonyms. For brevity, we may omit the $R$ and/or the $T$ when they are
  clear from the context.
\end{defn}

In \cite{Prz90}, Przymusinski introduced \emph{three-valued stable models},
which can be used to associate an LTS to a TSS with negative premises. Each
TSS has a least three-valued stable model, which coincides with the
well-founded semantics from \cite{GKRS91}. We consider the set of
\emph{sentences that are certainly true} in the least three-valued stable
model, which, for a TSS without negative premises, coincides with the set of
provable transitions in Definition~\ref{def:tss-induced-lts}. As Fokkink and
Vu noticed in \cite{FV03}, if $R$ is a TSS and $R'$ is obtained by removing
all the negative premises from the rules in $R$, then the LTS associated with
$R$ is included in the LTS associated with $R'$. In particular, if the LTS
associated with $R'$ has any of the bounded-nondeterminism properties
considered in this paper, then the LTS associated with $R$ has the property
too. We follow \cite{FV03} and ignore the negative premises in the TSSs. (Note
that ignoring negative premises makes the rule formats proposed in the paper
more restrictive, as negative premises may render some rules junk. However, as
in~\cite{FV03}, none of the rule formats that we introduce here impose any
restrictions on negative premises and the results we provide extend to a
setting with negative premises at the price of working with three-valued
stable models.) From now on we use `TSS with terms as labels' as a synonym of
`TSS with terms as labels and without negative premises'.

\begin{ntn}
  \label{not:sigma-zero}
  Many of the examples in the paper consider the following signature
  $\Sigma_0$ of reference. Let $L$ consist of countably many ground labels
  $l_i$ with $i\in\NN$ (\ie, the $l_i$ are constants). The signature
  $\Sigma_0$ consists of unary function symbols $f$ and $g$, and the constant
  symbols in $L$.
\end{ntn}

%%%%%%%%%%%%%%%%%%%%%%%%%%%%%%%%%%%%%%%%%%%%%%%%%%%%%%%%%%%%%%%%%%%%%%%%%%%
%% Dyadic transformation
\section{Dyadic transformations}
\label{sec:dyadic-transformation}
As mentioned in the introduction, our goal in this paper is to present a
unified rule format guaranteeing a variety of bounded nondeterminism
properties for TSSs with terms as labels. In order to do so, we introduce a
collection of dyadic transformations over TSSs that play a key role in our
approach.

The dyadic transformations stem from what we call the branching flow in
transitions (the notion of branching flow will become clear below and in
Definition~\ref{def:properties} in
Section~\ref{sec:other-properties}). Differently from the usual flow that
stresses the states in a sequence of transitions, the branching flow
emphasizes the one-to-many relation between the components of formulae in the
set of transitions of an LTS.

We introduce a family of transformations that turn the \emph{triadic} formula
representation $t\rel{l}t'$ into a \emph{dyadic} representation
$s\Rel r$ where $s$ is the \emph{source}, $r$ is the \emph{target} and $\Rel$
is a binary relation symbol in the set
$\{\rel{},\longleftarrow,\uparrow,\downarrow\}$. The relation symbol $\Rel$
does not carry any label, and the arbitrary term $l$ labelling the triadic
formula $t\rel{l}t'$ is part either of the source $s$ or the target $r$ in the
dyadic formula $s \Rel r$. Definition~\ref{def:dyadic-transformations} below
presents six dyadic transformations. Our dyadic representation is reminiscent 
of the \emph{commitments} of \cite{Mil93}, the \emph{residual data types} 
in Definition 13.9 of \cite{Ben10} and the \emph{residuals} in Definition 1 of \cite{PBEGW15}. 
In those three works, a transition consisting of an agent, an action and a derivative 
agent is represented as an agent and a commitment (resp. residual), which contains both 
the action and the derivative agent.

\begin{defn}[Dyadic transformation]
  \label{def:dyadic-transformations}
  Let $t\rel{l}t'$ be a triadic formula of a TSS with terms as labels. The
  dyadic transformations $\DT_k$ (with $1\leq k \leq 6$) are given by:
  \begin{center}
    \begin{minipage}{.5\linewidth}
      \begin{itemize}
      \item $\DT_1(t\rel{l}t')=t\rel{}(l,t')$.
      \item $\DT_2(t\rel{l}t')=t'\longleftarrow(l,t)$.
      \item $\DT_3(t\rel{l}t')=l\uparrow(t,t')$.
      \end{itemize}
    \end{minipage}
    \begin{minipage}{.48\linewidth}
      \begin{itemize}
      \item $\DT_4(t\rel{l}t')=(t,l)\rel{} t'$.
      \item $\DT_5(t\rel{l}t')=(t',l)\longleftarrow t$.
      \item $\DT_6(t\rel{l}t')=(t,t')\downarrow l$.
      \end{itemize}
    \end{minipage}
  \end{center}
  A transformation $\DT_k$ (with $1\leq k\leq 6$) is extended to a triadic TSS
  with terms as labels in the trivial way. Each rule
  \begin{mathpar}
    \inferrule*[right={\normalsize\qquad\qquad{\normalfont is transformed into}}]
    {\{u_i\rel{n_i}u'_i~|~i\in I\}}
    {t\rel{m}t'}
    \and
    % {\color{white}
    % \inferrule*
    % {{\color{black}\normalsize\text{is transformed into}}}{ }}
    % \and
    \inferrule*[right={\normalsize .}]
    {\{\DT_k(u_i\rel{n_i}u'_i)~|~i\in I\}}
    {\DT_k(t\rel{m}t')}
  \end{mathpar}
\end{defn}

The sets $\oSource[]{k}$ and $\oTarget[]{k}$ contain respectively the sources
and the targets over a signature $\Sigma$ in a $\DT_k$-transformed TSS. For
closed formulae $o\Rel{}d$ we use the terminology \emph{origin} for the $o$
and \emph{destination} for the $d$. The sets $\cSource[]{k}$ and
$\cTarget[]{k}$ contain respectively the origins (\ie, closed sources) and the
destinations (\ie, closed targets) over a signature $\Sigma$ in a
$\DT_k$-transformed TSS.
\begin{exmp}\label{Ex:source-origin-ex}
Applying the transformation $\DT_2$ to the rule 
  \begin{mathpar}
    \inferrule*[]
    {x\rel{y}x'}
    {f(x)\rel{f(y)}f(x')}
\end{mathpar}
yields the rule
  \begin{mathpar}
    \inferrule*[right={\normalsize .}]
    {x'\longleftarrow(y,x)}
    {f(x')\longleftarrow(f(y),f(x))} 
\end{mathpar}
The term $f(x')$ is the source of the transition formula
$f(x')\longleftarrow(f(y),f(x))$ and $(f(y),f(x))$ is its target. Therefore,
for all closed terms $p$ and $q$, the closed term $f(p)$ is an origin and
$(f(p),f(q))$ is a destination.
\end{exmp}
The dyadic transformations are duly extended to proof trees and LTSs. Notice
that the transformations $\DT_k$, with $1\leq k\leq 6$, are only
representational. Both a dyadic formula and its triadic analogue have type
$\oSource[]{k}\times\oTarget[]{k}=\oTerm\times\oTerm\times\oTerm$, \ie, both
formulae are structurally identical, and a triadic TSS and its dyadic counterpart
prove exactly the same set of transitions, in the sense of the following lemma.
\begin{lem}
  Let $R$ be a triadic TSS with terms as labels and let $R'$ be the dyadic TSS
  obtained by applying to $R$ some dyadic transformation $\DT_k$ with $1\leq
  k\leq 6$, \ie, $R'=\DT_k(R)$. A transition $p\rel{l}p'$ is provable in $R$
  iff $\DT_k(p\rel{l}p')$ is provable in $R'$.
\end{lem}
%\begin{proof}
%  Trivially, since $p\rel{l}p'$ and $\DT_k(p\rel{l}p')$ are structurally
%  identical.
%\end{proof}
% We say that an LTS in which the transitions are transformed to the
% dyadic representation is a dyadic LTS.

As we sketched before, the $\Rel$ in a dyadic TSS is a directed binary
relation symbol that reflects the branching flow. By way of example, consider
the dyadic transformation $\DT_1$. An origin $p$ that branches into a finite
number of destinations $(l,p')$ corresponds to a finitely branching process $p$
in the (triadic) LTS. An origin $p$ that branches into a possibly infinite number 
of destinations $(l,p')$ where there are only finitely many distinct $l$ corresponds 
to initials finiteness. On the other hand, in the setting of transformation $\DT_4$,
an origin $(p,l)$ that branches into a finite number of destinations $p'$ corresponds 
to image finiteness. (We defer a discussion of the correspondence of the other dyadic
transformations to other bounded-nondeterminism properties of interest to
Section~\ref{sec:other-properties}.)

The states in a sequence of transitions can be retrieved from the
relation symbol $\Rel$, by noticing the direction of the arrow that $\Rel$
depicts. (For $\uparrow$ and $\downarrow$, turn the arrow clockwise and
anti-clockwise respectively over the $(t,t')$.) For example, the
$\DT_6$-transformed transitions $(p_1,p_2)\downarrow l_1$ and
$(p_2,p_3)\downarrow l_2$, correspond to the sequence of triadic transitions
$p_1\rel{l_1}p_2\rel{l_2}p_3$.

For the first three dyadic transformations, \ie, $\DT_k$ where $1\leq k\leq
3$, the targets range over the Cartesian product $\oTerm\times\oTerm$. The
following notation will be adopted in what follows.
\begin{defn}
  \label{def:tr}
  Let $t\rel{l}t'$ be a formula of a triadic TSS and let $o \Rel (d_1,d_2)$ be
  the dyadic formula obtained by applying to $t\rel{l}t'$ one of the first
  three dyadic transformations $\DT_k$, \ie, $o \Rel
  (d_1,d_2)=\DT_k(t\rel{l}t')$ with $1\leq k\leq 3$.  We write
  $\tr_k(o_1,d_1,d_2)$ for the triadic transition obtained by assigning the
  roles of source, label and target to each of the $o$, $d_1$ and $d_2$
  according to $\DT_k$, in such a way that $\tr_k$ is the inverse of
  $\DT_k$. More precisely, $\tr_1(o,d_1,d_2)=o\rel{d_1}d_2$,
  $\tr_2(o,d_1,d_2)=d_2\rel{d_1}o$, and $\tr_3(o,d_1,d_2)=d_1\rel{o}d_2$.
%%%%%%%%%%%%%
%%  Let $R$ be a triadic TSS with terms as labels and let $R'$ be the dyadic TSS
%%  obtained by applying to $R$ some of the first three dyadic transformations
%%  $\DT_k$, \ie, $R'=\DT_k(R)$ with $1\leq k\leq 3$. Let $o$, $d_1$ and $d_2$
%%  be such that $o$ is the origin and $(d_1,d_2)$ is the destination of
%%  $R'$. We write $\tr_k(o_1,d_1,d_2)$ for the triadic transition in $R$
%%  obtained by assigning the roles of source, label and target to each of the
%%  $o$, $d_1$ and $d_2$ according to $\DT_k$, so $\tr_k$ is the inverse of
%%  $\DT_k$. More precisely, $\tr_1(o,d_1,d_2)=o\rel{d_1}d_2$,
%%  $\tr_2(o,d_1,d_2)=d_2\rel{d_1}o$, and $\tr_3(o,d_1,d_2)=d_1\rel{o}d_2$.
%%%%%%%%%%%%%%%%%%%%%%%%%%%%%%%%
\end{defn}

We consider the dyadic projections that are defined over the first three
dyadic transformations by projecting either on the first or on the second component
in the targets of transitions.
\begin{defn}[Dyadic projection]
  \label{def:dyadic-projections}
  Let $t\rel{l}t'$ be a formula of a triadic TSS. The dyadic projections
  $\DT_k^\prj$ (with $1\leq k \leq 3$ and $\prj$ either one of the projection
  functions $\pi_1$ and $\pi_2$) are given by:
  \begin{center}
    \begin{minipage}{.5\linewidth}
      \begin{itemize}
      \item $\DT_1^{\pi_1}(t\rel{l}t')=t\rel{}l$.
      \item $\DT_2^{\pi_1}(t\rel{l}t')=t'\longleftarrow l$.
      \item $\DT_3^{\pi_1}(t\rel{l}t')=l\uparrow t$.
      \end{itemize}
    \end{minipage}
    \begin{minipage}{.48\linewidth}
      \begin{itemize}
      \item $\DT_3^{\pi_2}(t\rel{l}t')=l\uparrow t'$.
      \item $\DT_2^{\pi_2}(t\rel{l}t')=t'\longleftarrow t$.
      \item $\DT_1^{\pi_2}(t\rel{l}t')=t\rel{} t'$.
      \end{itemize}
    \end{minipage}
  \end{center}
  A dyadic projection $\DT_k^{\prj}$ (with $1\leq k\leq 3$ and
  $\prj\in\{\pi_1,\pi_2\}$) is extended to a triadic TSS with terms as labels
  in the trivial way. Each rule
  \begin{mathpar}
    \inferrule*[right={\normalsize\qquad\qquad{\normalfont is transformed into}}]
    {\{u_i\rel{n_i}u'_i~|~i\in I\}}
    {t\rel{m}t'}
    \and
    \inferrule*[right={\normalsize .}]
    {\{\DT_k^{\prj}(u_i\rel{n_i}u'_i)~|~i\in I\}}
    {\DT_k^{\prj}(t\rel{m}t')}
  \end{mathpar}
\end{defn}
Notice that the dyadic projections $\DT_k^{\prj}$ are not only
representational transformations. Now, a projected formula has type
$\oSource[]{k}\times\prj(\oTarget[]{k})=\oTerm\times\oTerm$ (with $1\leq k\leq
3$ and $\prj\in\{\pi_1,\pi_2\}$), and its triadic analogue has type
$\oTerm\times\oTerm\times\oTerm$, which means that the two formulae are not
structurally identical. The dyadic projection entails a many-to-one relation
among the transitions of a triadic LTS and the transitions of its
dyadic-projected counterpart. That is, if $o\Rel d$ is a dyadic-projected
transition, then there exists a closed term $p$ such that $o$, $d$ and $p$ can
be assigned the roles of source, label and targets as to constitute a
transition in the original triadic LTS.

The rules of a triadic TSS with terms as labels are in a similar many-to-one
relation with the rules of the $\DT_k^{\prj}$-projected TSS. For illustration,
for each rule $\rho=\{v_i\rel{}w_i\mid i\in I\}/s\rel{}r$ in the
$\DT_1^{\pi_1}$-projected TSS there exist terms $t$ and $u_i$ with $i\in I$
such that $\{v_i\rel{w_i}u_i\mid i\in I\}/s\rel{r}t$ is a rule in the
original triadic TSS.
\begin{lem}
  Let $R$ be a triadic TSS with terms as labels and let $R'$ be the dyadic TSS
  obtained by applying to $R$ some dyadic projection $\DT_k^{\prj}$ (with
  $1\leq k\leq 3$ and $\prj\in\{\pi_1,\pi_2\}$), \ie, $R'=\DT_k^\prj(R)$. 
%%%%%%%%%%%%%%%
%%% The  following statements hold:
%%%  \begin{enumerate}[(i)]
%%%  \item 
Then, for each triadic transition $p\rel{l}p'$ that is provable in $R$, the
    transition $\DT_k^\prj(p\rel{l}p')$ is provable in $R'$.
%%%%%%%%%%%%%%%%%%%%%
%%%  \item For each transition $o\Rel d$ that is provable in $R'$ there exists a
%%%    closed term $q$ such that either $\tr_k(o,d,q)$ (if $\prj=\pi_1$) or
%%    $\tr_k(o,q,d)$ (if $\prj=\pi_2$) is provable in $R$.
%%  \end{enumerate}
\end{lem}
%%%%%%%% Omitted the proof %%%%%%%%%
\iffalse
\begin{proof}
  Statement $(i)$ is trivial by definition of $\DT_k^\prj$. For statement
  $(ii)$ we proceed by induction on the height of the proof tree of $o\Rel
  d$. The base case is when $o\Rel d$ is a substitution instance of an axiom
  $s\Rel r$ in $R'$. We only detail the proof when $\prj=\pi_1$, as the case
  $\prj=\pi_2$ is similar. By Definition~\ref{def:dyadic-projections} there
  exist a term $t$ and an axiom $\tr_k(s,r,t)$ in $R$ such that $\tr_k(o,d,p)$
  and $p$ are substitution instances (with the same substitution) of axiom
  $\tr_k(s,r,t)$ and term $t$ respectively, and the lemma holds.

  The general case is when $o\Rel d$ is a substitution instance of the
  conclusion of a rule $\rho'=H'/s\Rel r$ in $R'$ such that $o\Rel d$ has a
  proof tree $P'$ and $K'/o\Rel d$ is a substitution instance of $\rho'$, where
  $K'$ is the set of labels of the nodes directly above the root of $P'$. For
  every premise $v_i\Rel w_i\in H'$, there exists proof tree $Q'_i$ and, by the
  induction hypothesis, there are terms $u_i$ and proof trees $Q_i$ such that
  $Q_i$ proves a substitution instance of $\tr_k(v_i,w_i,u_i)$ for each
  $i$. By Definition~\ref{def:dyadic-projections} there exist term $t$ and
  rule $\rho=H/\tr_k(s,r,t)$ in $R$ such that $\tr_k(o,d,p)$ and $p$ are
  substitution instances (with the same substitution) of the conclusion
  $\tr_k(s,r,t)$ of rule $\rho$ and of term $t$ respectively, and the lemma
  holds.
\end{proof}
\fi
%%%%%%%%% End of omitted proof %%%%%%%%%%%%%
\begin{rem}
Let $R$ be a triadic TSS and let $R'$ be the dyadic TSS obtained by applying
to $R$ some dyadic projection $\DT_k^{\prj}$ (with $1\leq k\leq 3$ and
$\prj\in\{\pi_1,\pi_2\}$). Since dyadic projections abstract from information
that is present in the original triadic TSS, there might be some transition
$o\Rel d$ that is provable in $R'$, but for which there exists no closed term
$q$ such that either $\tr_k(o,d,q)$ (if $\prj=\pi_1$) or $\tr_k(o,q,d)$ (if
$\prj=\pi_2$) is provable in $R$. By way of example, consider the triadic TSS
$R_1$ with rules
  \begin{mathpar}
    \inferrule*[]
    {~}
    {c\rel{a}c'}
    \and
    \inferrule*[right={\normalsize .}]
    {x\rel{a}x}
    {f(x)\rel{a}x}
  \end{mathpar}
The dyadic TSS $R_1'$ obtained from $R_1$ using  $\DT_1^{\pi_1}$ is 
  \begin{mathpar}
    \inferrule*[]
    {~}
    {c\rel{}a}
    \and
    \inferrule*[right={\normalsize .}]
    {x\rel{}a}
    {f(x)\rel{}a}
  \end{mathpar}
The transition $f(c)\rel{}a$ is provable using the rules in $R'_1$. On the
other hand, there is no transition from $f(c)$ that is provable  in $R_1$. 

For what concerns the dyadic projection $\DT_1^{\pi_2}$, consider the triadic TSS
$R_2$ with rules
  \begin{mathpar}
    \inferrule*[]
    {~}
    {c\rel{a}c}
    \and
    \inferrule*[right={\normalsize .}]
    {x\rel{b}y}
    {f(x)\rel{b}f(y)}
  \end{mathpar}
The dyadic TSS $R_2'$ obtained from $R_2$ using  $\DT_1^{\pi_2}$ is 
  \begin{mathpar}
    \inferrule*[]
    {~}
    {c\rel{}c}
    \and
    \inferrule*[right={\normalsize .}]
    {x\rel{}y}
    {f(x)\rel{}f(y)}
  \end{mathpar}
The transition $f(c)\rel{}f(c)$ is provable using the rules in $R'_2$. On the
other hand, there is no transition from $f(c)$ that is provable  in $R_2$. 

Similar examples can be constructed for each of the other four dyadic
projections.
\end{rem}

\begin{ntn}
  In order to unify the notation, in what follows, we let
  $\prj\in\{\pi_1,\pi_2,\id\}$, where $\id$ is the identity function, and say
  `dyadic transformation $\DT_k^\prj$' when referring both to the dyadic
  transformations $\DT_k^\prj$ with $1\leq k\leq 6$ and $\prj=\id$, and to the
  dyadic projections $\DT_k^\prj$ with $1\leq k\leq 3$ and
  $\prj\in\{\pi_1,\pi_2\}$. We write $\oSource[\prj]{k}$, $\oTarget[\prj]{k}$,
  $\cSource[\prj]{k}$ and $\cTarget[\prj]{k}$ for the sets of sources,
  targets, origins, and destinations respectively in a $\DT_k^\prj$-dyadic
  LTS.

  From now on, we replace the binary relation symbol `$\Rel$' with the more
  suggestive `$\rel{}$' and write `$s\rel{} r$' instead of `$s\Rel r$' for any
  $\DT_k^\prj$-dyadic formula
\end{ntn}

\begin{defn}
  \label{def:DT-finite}
  Let $R$ be a triadic TSS with terms as labels and let $R'$ be the dyadic TSS
  obtained by applying to $R$ some dyadic transformation $\DT_k^\prj$, \ie,
  $R'=\DT_k^\prj(R)$.
  \begin{enumerate}[(i)]
  \item We say $R'$ is $\DT_k^\prj$-dyadic.
  \item We say $R$ is $\DT_k^\prj$-finite iff $R'$ is finitely branching, that
    is, iff for every origin $o$ in $R'$ the set $\{d\mid o\rel{} d\}$ is
    finite.
  \end{enumerate}
  (When the transformation $\DT_k^\prj$ is clear from the context, or when we
  refer to any of the dyadic transformations indistinctly, we will omit the
  tag `$\DT_k^\prj$-' and simply say `dyadic'.)
\end{defn}
For example, let $R$ be a triadic TSS with terms as labels. $R$ is
finitely branching iff $R$ is $\DT_1^\id$-finite. $R$ is image finite iff $R$ is
$\DT_4^\id$-finite. $R$ is initials finite iff $R$ is $\DT_1^{\pi_1}$-finite.

The dyadic transformations allow one to deal with SOS with higher-order
processes by using the framework for plain SOS. Notice that one can start with
a triadic TSS with terms as labels, apply a dyadic transformation, and then
trivially inject the obtained dyadic TSS into a triadic TSS with ground labels
by assuming that all transitions are labelled with the same constant term.

Next we explore the space of bounded nondeterminism properties that results
from Definition~\ref{def:DT-finite}.

%%%%%%%%%%%%%%%%%%%%%%%%%%%%%%%%%%%%%%%%%%%%%%%%%%%%%%%%%%%%%%%%%%%%%%%%%%%
%% Other properties
\section{Space of bounded-nondeterminism properties }
\label{sec:other-properties}

Definition~\ref{sec:other-properties} introduces a variety of
bounded-nondeterminism properties. Most of these properties do not hold for
reasonably expressive LTSs/process calculi. However, they have some
theoretical interest in the classification of the bounded-nondeterminism
properties.
\begin{defn}
  \label{def:properties}
  We define the following bounded-nondeterminism properties of a triadic LTS:
  \begin{enumerate}[(i)]
  \item Finite branching: $\forall p.~\{(l,p')\mid p\rel{l}p'\}\in
    \pset_\omega(\cTerm\times \cTerm)$.
  \item Finite folding: $\forall p'.~\{(p,l)\mid p\rel{l}p'\}\in
    \pset_\omega(\cTerm\times \cTerm)$.
  \item Finite bundling: $\forall l.~\{(p,p')\mid p\rel{l}p'\}\in
    \pset_\omega(\cTerm\times \cTerm)$.
  \item Image finiteness: $\forall p.\forall l.~\{p'\mid
    p\rel{l}p'\}\in\pset_\omega(\cTerm)$.
  \item Source finiteness: $\forall l.\forall p'.~\{p\mid
    p\rel{l}p'\}\in\pset_\omega(\cTerm)$.
  \item Label finiteness: $\forall p.\forall p'.~\{l\mid
    p\rel{l}p'\}\in\pset_\omega(\cTerm)$.
  \item Initials finiteness: $\forall p.~\{l\mid \exists p'.\
    p\rel{l}p'\}\in\pset_\omega(\cTerm)$.
  \item Finals finiteness: $\forall p'.~\{l\mid \exists p.\
    p\rel{l}p'\}\in\pset_\omega(\cTerm)$.
  \item Heads finiteness: $\forall l.~\{p\mid \exists p'.\
    p\rel{l}p'\}\in\pset_\omega(\cTerm)$.
  \item Tails finiteness: $\forall l.~\{p'\mid \exists p.\
    p\rel{l}p'\}\in\pset_\omega(\cTerm)$.
  \item Antecedents finiteness: $\forall p'.~\{p\mid \exists l.\
    p\rel{l}p'\}\in\pset_\omega(\cTerm)$.
  \item Consequents finiteness: $\forall p.~\{p'\mid \exists l.\
    p\rel{l}p'\}\in\pset_\omega(\cTerm)$.
  \end{enumerate}
\end{defn}
The properties listed above involve all the possible combinations of the three
components in a triadic formula (source, target and label) such that either
one component branches finitely into the other two, two components branch
finitely into the other one, or one component branches finitely into a second
component and possibly infinitely into the third
component. Figure~\ref{fig:properties} depicts these combinations
schematically. Finite branching $(i)$, image finiteness $(iv)$, and initials
finiteness $(vii)$ are well-known properties. We have contrived the names of
the rest of the properties, which, to the best of our knowledge, have not been
considered in the literature before.

Properties $(i)$ to $(xii)$ coincide with finite branching of the dyadic LTSs
obtained by applying each of the six dyadic transformations in
Definition~\ref{def:dyadic-transformations} and the six dyadic projections in
Definition~\ref{def:dyadic-projections}, respectively. In the nomenclature of
Definition~\ref{def:DT-finite}:
\begin{center}
\begin{minipage}{.5\linewidth}
  \begin{enumerate}[$(i)$]
  \item Finite branching $\defeq$ $\DT_1^{\id}$-finite.
  \item Finite folding $\defeq$ $\DT_2^{\id}$-finite.
  \item Finite bundling $\defeq$ $\DT_3^{\id}$-finite.
  \item Image finite $\defeq$ $\DT_4^{\id}$-finite.
  \item Source finite $\defeq$ $\DT_5^{\id}$-finite.
  \item Label finite $\defeq$ $\DT_6^{\id}$-finite.
  \end{enumerate}
\end{minipage}
\begin{minipage}{.48\linewidth}
  \begin{enumerate}[$(i)$]
    \addtocounter{enumi}{6}
  \item Initials finite $\defeq$ $\DT_1^{\pi_1}$-finite.
  \item Finals finite $\defeq$ $\DT_2^{\pi_1}$-finite.
  \item Heads finite $\defeq$ $\DT_3^{\pi_1}$-finite.
  \item Tails finite $\defeq$ $\DT_3^{\pi_2}$-finite.
  \item Antecedents finite $\defeq$ $\DT_2^{\pi_2}$-finite.
  \item Consequents finite $\defeq$ $\DT_1^{\pi_2}$-finite.
  \end{enumerate}
\end{minipage}
\end{center}

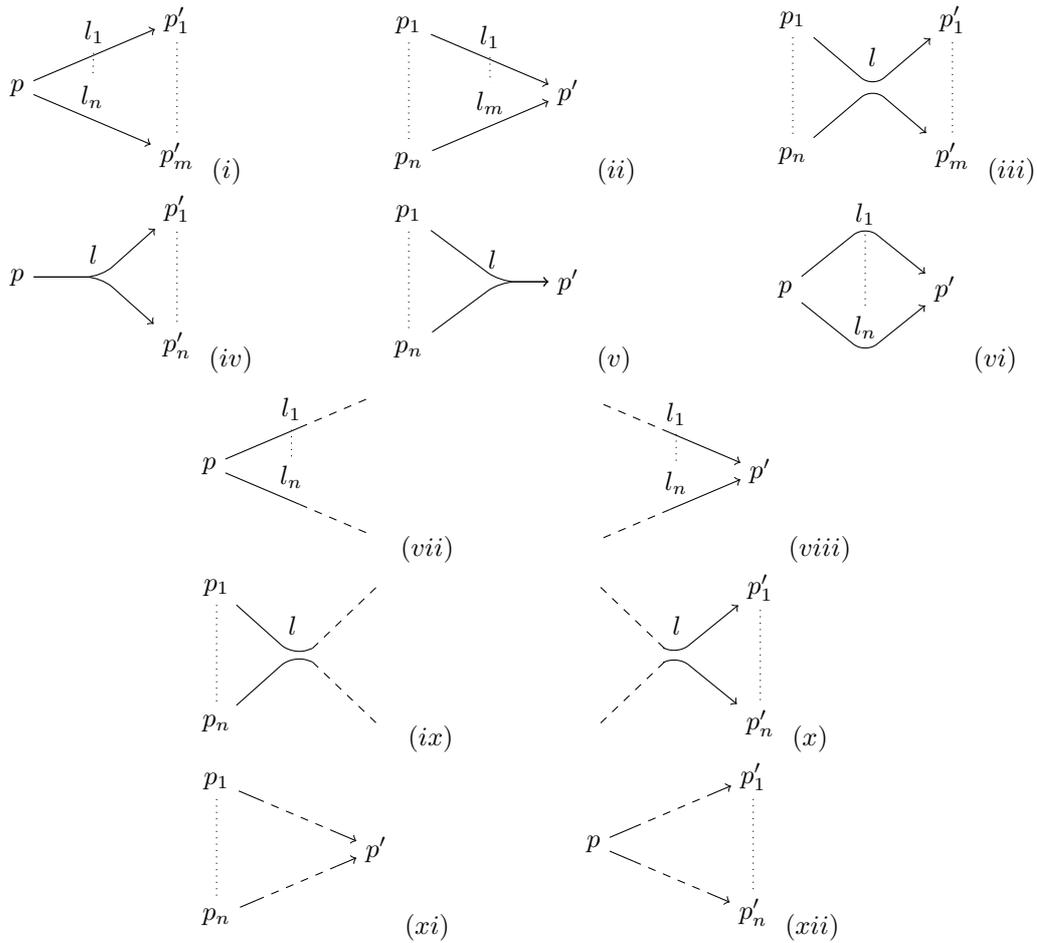
\begin{figure}
  \centering
  \begin{minipage}[b]{.3\linewidth}
  \begin{tikzpicture}[node distance=.9cm and 2.1cm]
    \node (p)  [on grid]                   {$p$};
    \node (p1) [on grid, above right=of p] {$p'_1$};
    \node (pn) [on grid, below right=of p] {$p'_m$};
    \path (p) edge [->] node (l1) [yshift=.3cm] {$l_1$} (p1);
    \path (p) edge [->] node (ln) [yshift=.3cm] {$l_n$} (pn);
    \path (p1) edge [dotted] (pn);
    \path (l1) edge [dotted] (ln);
  \end{tikzpicture}
  $(i)$
  \end{minipage}
  \begin{minipage}[b]{.3\linewidth}
  \begin{tikzpicture}[node distance=.9cm and 2.1cm]
    \node (p)                             {$p'$};
    \node (p1) [on grid, above left=of p] {$p_1$};
    \node (pn) [on grid, below left=of p] {$p_n$};
    \path (p1) edge [->] node (l1) [yshift=.3cm] {$l_1$} (p);
    \path (pn) edge [->] node (ln) [yshift=.3cm] {$l_m$} (p);
    \path (p1) edge [dotted] (pn);
    \path (l1) edge [dotted] (ln);
  \end{tikzpicture}
  $(ii)$
  \end{minipage}
  \begin{minipage}[b]{.3\linewidth}
  \begin{tikzpicture}[node distance=.9cm and 2.1cm]
    \node (p1) [on grid]             {$p_1$};
    \node (s)  [on grid,below=of p1] { };
    \node (pn) [on grid,below=of s]  {$p_n$};
    \node (l) at (1.05,-.5) {$l$};
    \node (p'1) [on grid,right=of p1] {$p'_1$};
    \node (p'n) [on grid,right=of pn] {$p'_m$};
    \draw [->,rounded corners=.2cm] (p1) to (1.05,-.9) to (p'1);
    \draw [->,rounded corners=.2cm] (pn) to (1.05,-.9) to (p'n);
    \path (p1) edge [dotted] (pn);
    \path (p'1) edge [dotted] (p'n);
  \end{tikzpicture}
  $(iii)$
  \end{minipage}
  \begin{minipage}[b]{.3\linewidth}
  \begin{tikzpicture}[node distance=.9cm and 2.1cm]
    \node (p)                     {$p$};
    \node (l) at (1,.3) {$l$};
    \node (p1) [on grid,above right=of p] {$p'_1$};
    \node (pn) [on grid,below right=of p] {$p'_n$};
    \draw [->,rounded corners=.2cm] (p) to (1.1,0) to (p1);
    \draw [->,rounded corners=.2cm] (p) to (1.1,0) to (pn);
    \path (p1) edge [dotted] (pn);
  \end{tikzpicture}
  $(iv)$
  \end{minipage}
  \begin{minipage}[b]{.3\linewidth}
  \begin{tikzpicture}[node distance=.9cm and 2.1cm]
    \node (p)                     {$p'$};
    \node (l) at (-1,.3) {$l$};
    \node (p1) [on grid,above left=of p] {$p_1$};
    \node (pn) [on grid,below left=of p] {$p_n$};
    \draw [->,rounded corners=.2cm] (p1) to (-.9,0) to (p);
    \draw [->,rounded corners=.2cm] (pn) to (-.9,0) to (p);
    \path (p1) edge [dotted] (pn);
  \end{tikzpicture}
  $(v)$
  \end{minipage}
  \begin{minipage}[b]{.3\linewidth}
  \begin{tikzpicture}[node distance=.9cm and 2.1cm]
    \node (l1) [on grid]             {$l_1$};
    \node (s)  [on grid,below=of l1] { };
    \node (b) [white,on grid,below=of s] {$p_n$};
    \node (ln) at (0,-1.5) {$l_n$};
    \node (p)  at (-1.05,-1)       {$p$};
    \node (p') at (1.05,-1)        {$p'$};
    \draw [->,rounded corners=.2cm] (p) to (0,-.15) to (p');
    \draw [->,rounded corners=.2cm] (p) to (0,-1.85) to (p');
    \path (l1) edge [dotted] (ln);
  \end{tikzpicture}
  $(vi)$
  \end{minipage}
  \begin{minipage}[b]{.3\linewidth}
  \begin{tikzpicture}[node distance=.9cm and 2.1cm]
    \node (p)  [on grid]                   {$p$};
    \node (p1) [white,on grid, above right=of p] {$p'_1$};
    \node (pn) [white,on grid, below right=of p] {$p'_n$};
    \path (p) edge [-] node (l1) [yshift=.45cm,xshift=.38cm]
                                 {$l_1$} (1.19,.51);
    \path (1.19,.51) edge [dashed] (2.1,.9);
    \path (p) edge [-] node (ln) [yshift=.15cm,xshift=.38cm]
                                 {$l_n$} (1.19,-.51);
    \path (1.19,-.51) edge [dashed] (2.1,-.9);
    \path (l1) edge [dotted] (ln);
  \end{tikzpicture}
  $(vii)$
  \end{minipage}
  \begin{minipage}[b]{.3\linewidth}
  \begin{tikzpicture}[node distance=.9cm and 2.1cm]
    \node (p)  [on grid]                   {$p'$};
    \node (p1) [white,on grid, above left=of p] {$p_1$};
    \node (pn) [white,on grid, below left=of p] {$p_n$};
    \path (p) edge [<-] node (l1) [yshift=.45cm,xshift=-.38cm]
                                 {$l_1$} (-1.19,.51);
    \path (-1.19,.51) edge [dashed] (-2.1,.9);
    \path (p) edge [<-] node (ln) [yshift=.15cm,xshift=-.38cm]
                                 {$l_n$} (-1.19,-.51);
    \path (-1.19,-.51) edge [dashed] (-2.1,-.9);
    \path (l1) edge [dotted] (ln);
  \end{tikzpicture}
  $(viii)$
  \end{minipage}\\
  \begin{minipage}[b]{.3\linewidth}
  \begin{tikzpicture}[node distance=.9cm and 2.1cm]
    \node (p1) [on grid]             {$p_1$};
    \node (s)  [on grid,below=of p1] { };
    \node (pn) [on grid,below=of s]  {$p_n$};
    \node (l) at (1,-.5) {$l$};
    \node (p'1) [white,on grid,right=of p1] {$p'_1$};
    \node (p'n) [white,on grid,right=of pn] {$p'_n$};
    \draw [-,rounded corners=.2cm] (p1) to (1,-.9) to (1.26,-.81);
    \draw [-,dashed] (1.26,-.81) to (2.1,0);
    \draw [-,rounded corners=.2cm] (pn) to (1,-.9) to (1.26,-.99);
    \draw [-,dashed] (1.26,-.99) to (2.1,-1.8);
    \path (p1) edge [dotted] (pn);
  \end{tikzpicture}
  $(ix)$
  \end{minipage}
  \begin{minipage}[b]{.3\linewidth}
  \begin{tikzpicture}[node distance=.9cm and 2.1cm]
    \node (p1) [white,on grid]             {$p_1$};
    \node (s)  [on grid,below=of p1] { };
    \node (pn) [white,on grid,below=of s]  {$p_n$};
    \node (l) at (1,-.5) {$l$};
    \node (p'1) [on grid,right=of p1] {$p'_1$};
    \node (p'n) [on grid,right=of pn] {$p'_n$};
    \draw [-,dashed] (0,0) to (.84,-.81);
    \draw [->,rounded corners=.2cm] (.84,-.81) to (1,-.9) to (p'1);
    \draw [-,dashed] (0,-1.8) to (.84,-.99);
    \draw [->,rounded corners=.2cm] (.84,-.99) to (1,-.9) to (p'n);
    \path (p'1) edge [dotted] (p'n);
  \end{tikzpicture}
  $(x)$
  \end{minipage}\\
  \begin{minipage}[b]{.3\linewidth}
  \begin{tikzpicture}[node distance=.9cm and 2.1cm]
    \node (p)                             {$p'$};
    \node (p1) [on grid, above left=of p] {$p_1$};
    \node (pn) [on grid, below left=of p] {$p_n$};
    \path (p1) edge [-] (-1.575,.675);
    \draw [dashed] (-1.575,.675) to (-.525,.225);
    \path (-.525,.225) edge [->] (p);
    \path (pn) edge [-] (-1.575,-.675);
    \draw [dashed] (-1.575,-.675) to (-.525,-.225);
    \path (-.525,-.225) edge [->] (p);
    \path (p1) edge [dotted] (pn);
  \end{tikzpicture}
  $(xi)$
  \end{minipage}
  \begin{minipage}[b]{.3\linewidth}
  \begin{tikzpicture}[node distance=.9cm and 2.1cm]
    \node (p)  [on grid]                   {$p$};
    \node (p1) [on grid, above right=of p] {$p'_1$};
    \node (pn) [on grid, below right=of p] {$p'_n$};
    \path (p) edge [-] (.525,.225);
    \draw [dashed] (.525,.225) to (1.575,.675);
    \path (1.575,.675) edge [->] (p1);
    \path (p) edge [-] (.525,-.225);
    \draw [dashed] (.525,-.225) to (1.575,-.675);
    \path (1.575,-.675) edge [->] (pn);
    \path (p1) edge [dotted] (pn);
  \end{tikzpicture}
  $(xii)$
  \end{minipage}
  \caption{Schematic representation of the properties in
    Definition~\ref{def:properties}. The dots indicate that the components
    branch finitely (\ie, the indices $n$ and the $m$ are natural
    numbers), and the dashes indicate that the missing components can branch
    infinitely.}
  \label{fig:properties}
\end{figure}

\begin{figure}[t]
  \centering
  \begin{tikzpicture}
    \matrix [row sep=.2cm, column sep=.25cm] {
      \node (i) [on grid] {$(i)$}; & & & & & &
      \node(iii)[on grid] {$(iii)$}; & & & & & &
      \node (ii) [on grid] {$(ii)$};
      & \node{derived};\\\\\\
      & & & & & \node (x) [on grid] {$(x)$}; & &
      \node (ix) [on grid] {$(ix)$};\\
      \node (xii) [on grid] {$(xii)$};
      & & & & & & & & & & & & \node (xi) [on grid] {$(xi)$};
      & \node{complementary};\\
      & & & & & \node (vii) [on grid] {$(vii)$}; & &
      \node (viii) [on grid] {$(viii)$};\\\\\\
      \node (iv) [on grid] {$(iv)$}; & & & & & &
      \node (vi) [on grid] {$(vi)$}; & & & & & &
      \node (v) [on grid] {$(v)$};
      & \node{elementary};\\
    };
    \path (i) edge (vii);
    \path (vii) edge (vi);

    \path (i) edge (xii);
    \path (xii) edge (iv);

    \path (iii) edge (x);
    \path (x) edge (iv);

    \path (iii) edge (ix);
    \path (ix) edge (v);

    \path (ii) edge (viii);
    \path (viii) edge (vi);

    \path (ii) edge (xi);
    \path (xi) edge (v);
  \end{tikzpicture}
  \caption{Hasse diagram of the properties in
    Definition~\ref{def:properties}.}
  \label{fig:hasse}
\end{figure}
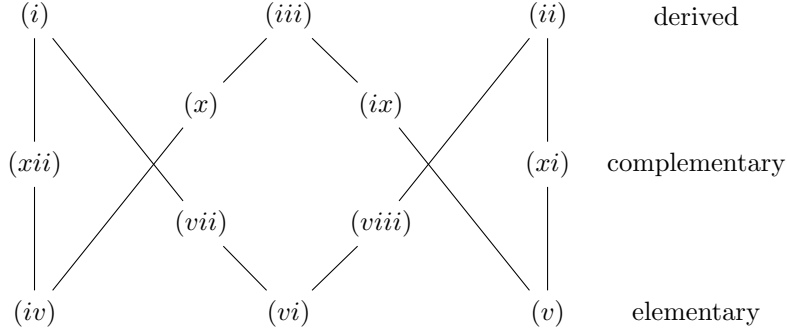

Properties $(i)-(xii)$ can be arranged in the Hasse diagram of Figure~\ref{fig:hasse},
which depicts the information order of the properties, that is, the order
relation is reverse implication. We write $q\leq p$ if $q$ occurs below $p$ in
the Hasse diagram. Properties $(iv)$, $(vi)$ and $(v)$ at the bottom are
elementary (\ie, most general) and properties $(i)$, $(iii)$ and $(ii)$ at the
top are derived (\ie, most specific). Properties $(vii)$ to $(xii)$ are
complementary (\ie, the difference between a derived and an elementary, more
on this below).

%Let us picture the Hasse diagram as a closed curve with three lobes. 
The properties are displayed in three layers: the bottom layer
contains three elementary properties; the top layer contains three derived
properties; the middle layer contains six complementary properties, each of
them in the middle of a path from an elementary to a derived property.

\begin{prop}
  \label{pr:covers}
  For every two bounded-nondeterminism properties $p$ and $q$ of
  Definition~\ref{def:properties}, $p$ implies $q$ iff $q\leq p$ in the Hasse
  diagram of Figure~\ref{fig:hasse}.
\end{prop}
\begin{proof}
  Recall notation $\tr_k$ from Definition~\ref{def:tr}. For the
  ($\Longleftarrow$) direction, a derived property $d$ that corresponds to a
  dyadic transformation $\DT_k^\id$ (with $1\leq k\leq3$) states that a
  component $t_1$ branches finitely into the other components $t_2$ and $t_3$,
  \ie, for each $t_1$, the set $\{(t_2,t_3)\mid \tr_k(t_1,t_2,t_3)\}$ is
  finite. Trivially, this implies the two complementary properties $c_1$ and
  $c_2$ that state that the $t_1$ is respectively allowed to branch infinitely
  into either the $t_2$ or the $t_3$, \ie, for each $t_1$ both the sets
  $\{t_2\mid\exists t_3.\tr_k(t_1,t_2,t_3)\}$ and $\{t_3\mid\exists
  t_2.\tr_k(t_1,t_2,t_3)\}$ are finite. In turn, the complementary properties
  $c_1$ and $c_2$ respectively imply the two elementary properties $e_1$ and
  $e_2$ that state that, given a $t_1$, the $t_2$ and $t_3$ branch finitely
  into each other, \ie, for each $t_1$ and $t_2$ the set $\{t_3\mid
  \tr_k(t_1,t_2,t_3)\}$ is finite, and for each $t_1$ and $t_3$ the set
  $\{t_2\mid \tr_k(t_1,t_2,t_3)\}$ is finite.

  For the ($\Longrightarrow$) direction, no implications other than the ones
  given by the Hasse diagram of Figure~\ref{fig:hasse} exist. This can be
  shown straightforwardly by picking counter-examples for each case, which we
  omit here.
\end{proof}
For instance, properties $(xii)$ and $(vii)$ are below property $(i)$ in the
diagram, and, in turn, properties $(iv)$ and $(vi)$ are below properties
$(xii)$ and $(vii)$ respectively. Therefore, $(i)$ implies each of $(xii)$,
$(vii)$, $(iv)$ and $(vi)$.

Some equivalences between combinations of the properties depicted in the Hasse
diagram of Figure~\ref{fig:hasse} can be established.
\begin{prop}
  \label{pr:derived-equivalent}
  A bounded-nondeterminism derived property $d$ is equivalent to the
  conjunction of the elementary property $e$ from one of the chains that has
  $d$ as maximum in the Hasse diagram of Figure~\ref{fig:hasse}, and the
  complementary property $c$ from the adjacent chain that also has $d$ as
  maximum, \ie, $d\Longleftrightarrow e \land c$.
\end{prop}
\begin{proof}
  By Definition~\ref{def:properties} the derived property $d$ corresponds to a
  dyadic transformation $\DT_k^\id$ with $1\leq k\leq 3$. We only detail the
  proof when the complementary property $c$ corresponds to $\DT_k^{\pi_1}$, as
  the case that corresponds to $\DT_k^{\pi_2}$ is similar. The
  ($\Longrightarrow$) direction is trivial by Proposition~\ref{pr:covers}. For
  the ($\Longleftarrow$) direction, notice that the elementary property $e$
  states that two components $t_1$ and $t_2$ branch finitely into the other
  component $t_3$, \ie, for each $t_1$ and $t_2$ the set $\{t_3\mid
  \tr_k(t_1,t_2,t_3)\}$ is finite. The complementary property $c$ from the
  other chain states that component $t_1$ branches finitely into component
  $t_2$, but possibly branches infinitely into component $t_3$. By applying
  property $e$ to each of the finitely many pairs $(t_1,t_2)$, component $t_1$
  is shown to branch finitely into component $t_3$, and thus the derived
  property $d$ holds, \ie, for each $t_1$ the set $\{(t_2,t_3)\mid
  \tr_k(t_1,t_2,t_3)\}$ is finite.
\end{proof}
For instance, properties $(iv)$ and $(vii)$ together are equivalent to $(i)$,
and so are properties $(vi)$ and $(xii)$ together.

\begin{prop}
  A bounded-nondeterminism derived property $d$ is equivalent to the
  conjunction of the two complementary properties $c_1$ and $c_2$ from the two
  adjacent chains that have $d$ as maximum.
\end{prop}
\begin{proof}
  The ($\Longrightarrow$) direction is trivial by
  Proposition~\ref{pr:covers}. For the ($\Longleftarrow$) direction, by
  Proposition~\ref{pr:covers} the complementary property $c_i$ (with $i\in
  \{1,2\}$) implies the elementary property $e_i$ such that $e_i\leq c_i$, and
  by Proposition~\ref{pr:derived-equivalent}, $e_i\land c_j$ with $i,j\in
  \{1,2\}$ and $i\not=j$ implies $d$.
\end{proof}
A complementary property in a chain that has a derived property $d$ as maximum
is the difference between $d$ and the elementary property from the adjacent
chain that also has $d$ as maximum. Notice that in order to imply a derived
property, at least one complementary property from either of the adjacent
chains that have the derived one as maximum must take part in the
conjunction. The following proposition illustrates this.
\begin{prop}
  \label{pr:two-elementary}
  A bounded-nondeterminism derived property $d$ is not implied by the
  conjunction of the two elementary properties $e_1$ and $e_2$ that are below
  $d$ in the Hasse diagram of Figure~\ref{fig:hasse}.
\end{prop}
\begin{proof}
  By Definition~\ref{def:properties} the derived property $d$ corresponds to a
  dyadic transformation $\DT_k^\id$ with $1\leq k\leq 3$. The derived property
  $d$ requires that one component $t_1$ branches finitely into the components
  $t_2$ and $t_3$, \ie, for every $t_1$, the set $\{(t_2,t_3)\mid
  \tr_k(t_1,t_2,t_3)\}$ is finite. The elementary properties $e_1$ and $e_2$
  that are below $d$ respectively state that, given a $t_1$, the $t_2$ and
  $t_3$ branch finitely into each other, \ie, for every $t_1$ and $t_2$ the
  set $\{t_3\mid \tr_k(t_1,t_2,t_3)\}$ is finite, and for every $t_1$ and
  $t_3$ the set $\{t_2\mid \tr_k(t_1,t_2,t_3)\}$ is finite. But for each
  $t_1$, the $t_2$ and the $t_3$ could be infinitely many and still branch
  finitely into each other.

  This idea can be exploited to construct counter-examples to the implication
  $e_1\land e_2\implies d$. By way of example, consider the triadic LTS $L$
  given by $\{p_0\rel{l_i}p_i\mid i\in\NN\}$. $L$ is image finite (property
  $(iv)$ in Figure~\ref{fig:hasse}) because for each $p_i$ and $l_j$ there is
  at most one $p_\ell$ such that $p_i\rel{l_j}p_\ell$ (\ie, $i=0$, $j=\ell$),
  and $L$ is labels finite (property $(vi)$ in Figure~\ref{fig:hasse}) because
  for each $p_i$ and $p_\ell$ there is at most one $l_j$ such that
  $p_i\rel{l_j}p_\ell$ (\ie, $i=0$, $j=\ell$). However, $L$ is not finitely
  branching because for $p_0$ there are infinitely many pairs $(l_i,p_i)$ such
  that $p_0\rel{l_i}p_i$.
\end{proof}

For further illustration, the equivalence class of property $(i)$ up to
logical equivalence is
% $\{(i), (xii) \land (vi), (vii) \land (iv), (xii) \land (vii)\}$, and property
% $(i)$ is implied by any of the properties in $\{(xii) \land (viii), (xii)
% \land (ii), (vii) \land (x), (vii) \land (iii)\}$.
\begin{displaymath}
  \{(i), (xii) \land (vi), (vii) \land (iv), (xii) \land (vii)\},
\end{displaymath}
and property $(i)$ is implied by any of the properties in
\begin{displaymath}
  \{(xii) \land (viii), (xii) \land (ii), (vii) \land (x), (vii) \land (iii)\}.
\end{displaymath}

In the light of the dyadic transformations, the rule format in \cite{FV03} for
finite branching is enough to guarantee $\DT_k^\prj$-finiteness of the LTS
associated to a triadic TSS with terms as labels. That is, by applying each of
the dyadic transformations as a pre-processing step, that rule format ensures
all the derived, elementary and complementary bounded-nondeterminism
properties in this section.

Next we show how to improve on the rule format of \cite{FV03} to cover more
TSSs than the ones covered there. Sections~\ref{sec:junk-rules-eta-types} and
\ref{sec:refining-eta-types} explain how the set of finitely branching TSSs
that meet the conditions of the rule format can be enlarged by considering
junk rules (\ie, rules that cannot give rise to transitions), and by refining
the $\eta$-types of \cite{FV03} such that some junk rules do not have
$\eta$-type. Section~\ref{sec:s-types} introduces the $S$-types, which adapt
these refinements on the $\eta$-types to dyadic TSSs. Finally,
Section~\ref{sec:rule-formats} presents rule formats that are based on the
$S$-types and guarantee $\DT_k^\prj$-finiteness of the LTS associated with a
triadic TSS with terms as labels.

%%%%%%%%%%%%%%%%%%%%%%%%%%%%%%%%%%%%%%%%%%%%%%%%%%%%%%%%%%%%%%%%%%%%%%%%%%%
%% Junk rules
\section{Junk rules and the $\eta$-types}
\label{sec:junk-rules-eta-types}
A rule format is a set of conditions on the rules of a TSS that ensure that
the associated LTS enjoys some property of interest (\eg, being finitely
branching). A rule format imposes \emph{sufficient conditions} on the TSS
for the associated LTS to enjoy the property. In devising a rule format there is
a trade-off between generality and ease of use: the syntactic conditions typically
do not cover all the TSSs whose associated LTSs have the property of interest,
but only a reasonably large class of those that do so. For instance, a TSS might 
contain rules that do not give rise to transitions (the so-called junk rules) and do 
not fit the rule format. In general, the conditions of 
a rule format could be weakened in order to enlarge the class of TSSs that are covered
by it and enjoy the property of interest. However, this usually leads to more complex conditions.

One of our goals in this paper is to improve the rule format for bounded nondeterminism in \cite{FV03} so that
it can filter out more junk rules. To do so, in this section we will explain the
notion of $\eta$-type from \cite{FV03} and how it can be used to filter out some junk rules. 
Next, in Section~\ref{sec:refining-eta-types}, we will refine the notion of $\eta$-type from \cite{FV03} 
in order to increase the number of TSSs that our rule format (presented in Section~\ref{sec:rule-formats}) covers.

\begin{note}
  The following considerations about junk rules are orthogonal to the dyadic
  transformations that we have presented before. Therefore, for the sake of
  clarity in the presentation, in the remainder of the section and in
  Section~\ref{sec:refining-eta-types} we consider triadic TSSs with ground
  labels.

  Familiarity with the rule format for finite branching from \cite{FV03} will
  help in appreciating this section. However, we have striven to make the
  presentation understandable even with no prior knowledge of that paper.
\end{note}
%%%%%%%%%%%%%%%%%%%%%%%
Intuitively, the rule format for bounded nondeterminism in~\cite{FV03} places
restrictions on the allowed rules ensuring that, for each closed term $p$,
\begin{enumerate}
\item the rules in the TSS do not allow one to simulate `unguarded recursion'
  for $p$,
\item only finitely many rules can be employed to derive transitions from $p$,
and
\item each rule can only be used to infer finitely many transitions from $p$.
\end{enumerate}
The third property is checkable for each rule in isolation and is embodied in
the requirement that the TSS be in bounded nondeterminism format (see
Definition~\ref{def:bounded-nondeterminism-format} to follow). This format
essentially ensures that the transitions of $p$ that can be derived from a
rule depend only on the transitions of the closed terms that are substituted
for the variables that occur in the source of the rule. For example, rules of
the form
  \begin{mathpar}
    \inferrule*[]
    {~}
    {f(x) \rel{a} y}
    \and
    % {\color{white}
    % \inferrule*
    % {{\color{black}\normalsize\text{is transformed into}}}{ }}
    % \and
    \inferrule*[right={\normalsize .}]
    {y\rel{a}y'}
    {f(x)\rel{a}y'}
  \end{mathpar}
are \emph{not} in bounded nondeterminism format because the instantiation of
the variable $y$ can be chosen arbitrarily and thus its behaviour need not
have any connection with that of the closed term that is substituted for the
variable $x$. It is easy to see how rules of that form can be used to derive
infinitely many transitions from a term of the form $f(p)$. 

On the other hand, the first and the second properties are `global' and need
to be checked for sets of rules. The existence of a strict stratification (see
Definition~\ref{def:partial-strict-stratification}) enforces the first
property, while the second is guaranteed by the requirement that the TSS be
bounded (see Definition~\ref{def:bounded-new} for our version of that
notion). In order to define the notion of bounded TSS, Fokkink and Vu classify
the transition rules in a TSS according to their so-called
$\eta$-types. 

The intuition behind the use of the $\eta$-types, presented in Definition~11
of \cite{FV03}, is that they enforce that the number of rules that give rise
to transitions from a given process is finite.  The $\eta$-types are
parametric on a map $\eta:\oTerm\to\pset_\omega(\oTerm)$ that specifies a
predefined, finite set of terms $\eta(t)$ for each \mbox{source $t$}. For the
purpose of this discussion, our readers can think of $\eta(t)$ as being the
set of the sources of premises in rules having $t$ as the source of the
conclusion; intuitively, those are the terms on which initial transitions from
instantiations of $t$ may depend. By way of example, let $A$ be a countably
infinite set of actions, and consider the (infix) binary operator $+$ from CCS
in \cite{Mil89} with rules \label{ex:example-1}
  \begin{mathpar}
    \inferrule*[left=L$_a$,right={\normalsize $,\qquad a\in A$}]
    {x\rel{a}x'}
    {x+y\rel{a}x'}
    \and
    \inferrule*[left=R$_a$,right={\normalsize $,\qquad a\in A.$}]
    {y\rel{a}y'}
    {x+y\rel{a}y'}
  \end{mathpar}
The initial behaviour of an instantiation of $x+y$ depends only on that of the
terms (substituted for) $x$ and $y$. Thus, based on what we said earlier, it
is natural to define $\eta(x+y)=\{x,y\}$. We say that $\eta$ is the
\emph{support map} and call $\eta(t)$ the \emph{support for source $t$}. 

Since $A$ is infinite, the set of rules having $x+y$ as source is also
infinite. However, when the closed terms $p$ and $q$ are finitely branching,
only finitely many rules can be used to derive transitions from $p+q$. The
technical notion of $\eta$-type is meant to capture this fact. Two rules with
the same source $t$ have different $\eta$-type if, and only if, for some term
$u$ in the support for $t$, the rules differ in the respective sets of labels
that appear in their premises with source $u$. An $\eta$-type is written
$\langle t,\psi\rangle$, where $t$ is the source of the rule and
$\psi:\eta(t)\to\pset_\omega(\Label)$ is the map that delivers for each term
$u$ in the support for $t$, a finite set of labels. Intuitively, $\psi(u)$ is
the set of labels in the premises with source $u$.

Continuing with our example, the $\eta$-type of rule \textsc{L$_a$} is
$\langle x+y,\{x\mapsto\{a\},y\mapsto\emptyset\}\rangle$, and the $\eta$-type
of rule \textsc{R$_a$} is $\langle
x+y,\{x\mapsto\emptyset,y\mapsto\{a\}\}\rangle$. Note that different rules
have different $\eta$-types.
%%%%%%%%%%%%%%%%%%%%%%%%%%%%%%%%
%%%%%%%%%%% The message of this example is now in the expanded
%%%%%%%%%%% running text. 
\iffalse
\begin{exmp}
\label{ex:example-1}
  Let $A$ be a countably infinite set of actions. Consider the (infix) binary operator
  $+$ from CCS in \cite{Mil89} with rules
  \begin{mathpar}
    \inferrule*[left=L$_a$,right={\normalsize $,\qquad a\in A$}]
    {x\rel{a}x'}
    {x+y\rel{a}x'}
    \and
    \inferrule*[left=R$_a$,right={\normalsize $,\qquad a\in A.$}]
    {y\rel{a}y'}
    {x+y\rel{a}y'}
  \end{mathpar}
  Let $\eta(x+y)=\{x,y\}$. The $\eta$-type of rule \textsc{L$_a$} is $\langle
  x+y,\{x\mapsto\{a\},y\mapsto\emptyset\}\rangle$, and the $\eta$-type of rule
  \textsc{R$_a$} is $\langle x+y,\{x\mapsto\emptyset,y\mapsto\{a\}\}\rangle$.
\end{exmp}
\fi
%%%%%%%%%%%%%%%%%%%%%%%%%%%%%%%
The $\eta$-types of the rules in a TSS that meet the conditions of the rule
format of \cite{FV03} are required to be finitely inhabited. That is, for each
$\eta$-type $\langle t,\psi\rangle$ there are only finitely many rules having
$\langle t,\psi\rangle$ as their $\eta$-type. This constraint is met by the
$\eta$-types of the rules for the $+$ operation.

The rule format from \cite{FV03} requires that the images in the codomain of
both maps $\eta$ and $\psi$ be finite. The requirement that, for each term
$t$, the support $\eta(t)$ is finite is a necessary condition for the
associated LTS to be finitely branching. The following example illustrates
why.
\begin{exmp}
  Consider the following TSS over the signature $\Sigma_0$ of
  Notation~\ref{not:sigma-zero}, where $L= \{l_i\mid i\in\NN\}$:
  \begin{mathpar}
    \inferrule*[right={\normalsize $,\qquad i\in\NN$}]
    { }
    {l_i\rel{l_i}l_i}
    \and
    \inferrule*[right={\normalsize $,\qquad i\in\NN$.}]
    {l_i\rel{l_i}l_i}
    {f(x)\rel{l_i}l_i}
  \end{mathpar}
  Notice that a source $f(p)$ (with $p\in\cTerm$) affords infinitely many
  transitions ${f(p)\rel{l_i}l_i}$ (with $i\in\NN$) and the associated LTS is
  infinitely branching. However, if we allowed the sets in the codomain of
  $\eta$ to be infinite, then we could have $\eta(f(x))=L$ and the $i$th
  instantiation of the rule template on the right would have $\eta$-type
  $\langle f(x),\psi_i\rangle$, where $\psi_i$ is such that
  $\psi_i(l_i)=\{l_i\}$, and $\psi_i(l_j)=\emptyset$ if $j\neq i$.  The
  $\eta$-types would be finitely inhabited and the rule format would not
  guarantee finite branching of the associated LTS.

  The requirement that $\eta(f(x))$ be finite enforces that the infinitely
  many instantiations of the rule template on the right, which all share the
  same source $f(x)$, must be distinguished by focusing only on finitely many
  sources of their premises. Since the set of possible sources of their
  premises is infinite, \ie, $\{l_i\mid i\in\NN\}$, it is actually impossible
  to distinguish all of them, and the instantiations whose premises do not
  have their source in the finite support set $\eta(f(x))$ have one and the same
  \mbox{$\eta$-type} $\langle f(x),\{t\mapsto\emptyset\mid
  t\in\eta(f(x))\}\rangle$, which is infinitely inhabited. Suitably, the TSS
  does not meet the conditions of the rule format from \cite{FV03}.
\end{exmp}

On the contrary, the requirement that, for each rule with source $t$ and for
each term $u\in\eta(t)$, the set of labels in the premises with source $u$ is
finite is not a necessary condition for the associated LTS to be finitely
branching.
%%%%%%%%%%%%%%%%%
%%\footnote{We noticed this fact in our previous work \cite{AcetoGI16} and
%%  removed the requirement from our rule format there. Alas, including the
%%  requirement \emph{is beneficial} for the rule format to cover more
%%  cases. Here we rectify and reinstate the requirement as originally in
%%  \cite{FV03}.} 
%%%%%%%%%%%%%%%%%
However, the rule format of \cite{FV03} would report more
false negatives if this requirement were omitted. That is, more TSSs would be
considered as not satisfying the conditions for finite branching although
their associated LTSs are finitely branching. The reason is the presence of a
certain class of junk rules. Again, let us illustrate this point with an example.
\begin{exmp}
  \label{ex:infinite-premises}
  Consider the following TSS over the signature $\Sigma_0$ of
  Notation~\ref{not:sigma-zero}, where $L= \{l_i\mid i\in\NN\}$:
  \begin{mathpar}
    \inferrule*[]
    { }
    {l_1\rel{l_1}l_1}
    \and
    \inferrule*[right={\normalsize $,\qquad i\in\NN$.}]
    {\{l_1\rel{l_j}l_j\mid j\in \NN\}}
    {f(l_1)\rel{l_i}l_i}
  \end{mathpar}
  The associated LTS is finitely branching because the only provable transition
  is ${l_1\rel{l_1}l_1}$. The instantiations of the rule template on the right
  have infinitely many premises with source $l_1$ and label and target $l_j$
  where $j\in\NN$. However, none of the instantiations can take part in a
  proof tree because the infinitely many distinct premises share source $l_1$
  and the transition $l_1\rel{l_1}l_1$ cannot prove all of them (\ie, all the
  instantiations are junk rules). Assume ${\eta(f(l_1))=\{l_1\}}$. Happily,
  none of the instantiations of the rule template on the right has a valid
  \mbox{$\eta$-type}, because no $\psi$ exists such that the finite set
  $\psi(l_1)$ contains every label in the premises with source $l_1$ (\ie,
  $\psi(l_1)\not=\{l_j\mid j\in\NN\}$). Therefore the $\eta$-types are
  finitely inhabited and the TSS meets the conditions of the rule format from
  \cite{FV03}. On the other hand, if $\psi(l_1)$ were allowed to be equal to
  $\{l_j\mid j\in\NN\}$, then all the instantiations of the rule template on
  the right would have the same $\eta$-type $\langle f(x),\{l_1\mapsto
  L\}\rangle$, which would be infinitely inhabited.
\end{exmp}

%%%%%%%%%%%%%%%%%%%%%%%%%%%%%%%%%%%%%%%%%%%%%%%%%%%%%%%%%%%%%%%%%%%%%%%%%%%
%% Refining eta-types
\section{Refining the $\eta$-types}
\label{sec:refining-eta-types}

Although the rule format of \cite{FV03} covers TSSs such as the one in
Example~\ref{ex:infinite-premises}, it falls short of covering TSSs with
other kinds of junk rules. In the following subsections we analyze some of
these cases and provide successive refinements to the $\eta$-types of
\cite{FV03} that we will incorporate in the rule format we present in
Section~\ref{sec:rule-formats}.
% such that the resulting rule format covers these cases.

\subsection{Restricted support}
\label{sec:restricted-support}
A rule of a TSS whose premises do not unify with conclusions of other rules
in the TSS is certainly a junk rule, since its premises will never be
true. But this fact is not always ascertained by not having a valid
$\eta$-type, since the $\eta$-types rely only on local information. The
following example illustrates this observation.
\begin{exmp}
  \label{ex:restricted-support}
  Consider the following TSS over the signature $\Sigma_0$ of
  Notation~\ref{not:sigma-zero}:
  \begin{mathpar}
    \inferrule*[]
    { }
    {g(l_1)\rel{l_1}l_1}
    \and
    \inferrule*[right={\normalsize $,\quad\quad i\in\mathbb{N}$.}]
    {g^i(x)\rel{l_i}x}
    {f(x)\rel{l_i}x}
  \end{mathpar}
  Here $g^i$ stands for applying the unary function symbol $g$ to its argument
  $i$ times. Notice that the only provable transitions are
  $g(l_1)\rel{l_1}l_1$ and $f(l_1)\rel{l_1}l_1$, and therefore the TSS induces
  an LTS that is finitely branching. Assume that the TSS has a support map such
  that
  \begin{displaymath}
    \begin{array}{rcll}
      \eta(g(l_1))&=&\emptyset , \\
      \eta(f(x))&=&\{g(x)\} &\text{and}\\
      \eta(t)&=&\emptyset&\text{\ow{.}}
    \end{array}
  \end{displaymath}
  Clearly, all the instantiations of the rule template on the right with
  $i\not=1$ are junk rules because their premises do not unify with the
  conclusion of any rule in the TSS. However, all these instantiations have
  $\langle f(x),\emptyset\rangle$ as their \mbox{$\eta$-type}, which is
  infinitely inhabited. (Notice that no $\eta$ exists such that the
  $\eta$-types would be finitely inhabited.) Thus the TSS does not meet the
  conditions of the rule format of \cite{FV03}, even though the associated LTS is finitely branching.
\end{exmp}

\begin{obs}
  \label{ob:restricted-support}
  {\it The $\eta$-types have to be refined so that the support map $\eta$ be
    restricted to deliver only sets of terms that unify with sources of rules
    in the TSS. According to this, a rule with source $t$ should only have a
    valid $\eta$-type if the set of sources of its premises is a subset of
    the restricted support $\eta(t)$.}
\end{obs}

Take the support map $\eta$ of Example~\ref{ex:restricted-support} above. That
function $\eta$ meets the condition of the restricted support because every
term in a set delivered by $\eta$ (\ie, $g(x)\in\eta(f(x))$) unifies with some
rule in the TSS (\ie, $\sigma(g(x))=g(l_1)$, with $\sigma=\{x\mapsto
l_1\}$). The instantiations of the rule template on the right with $i\not=1$
should not have a valid $\eta$-type because $g^i(x)\not\in\eta(f(x))$ when
$i\neq 1$.

\subsection{Targets of premises are important too}
\label{sec:targets}
Two rules with the same source $t$ have different $\eta$-types if there is some term in $\eta(t)$
for which the two rules have premises with different actions. For instance, rules \textsc{L$_a$} and \textsc{R$_a$}
on page~\pageref{ex:example-1} have different $\eta$-types because the only premise in \textsc{L$_a$} is $y\rel{a}y'$, whereas  
the only premise in \textsc{R$_a$} is $x\rel{a}x'$.

Sometimes it is convenient to distinguish rules with the same source
not only by focusing on the actions of their premises but also by looking at
the targets of their premises. This is shown in the following example.
\begin{exmp}
  \label{ex:targets}
  Consider the following TSS over the signature $\Sigma_0$ of
  Notation~\ref{not:sigma-zero}:
  \begin{mathpar}
    \inferrule*[left=1]
    { }
    {g(l_1)\rel{l_1}l_1}
    \and
    \inferrule*[left=2$i$,right={\normalsize $,\quad\quad i\in\mathbb{N}$.}]
    {g(x)\rel{l_1}l_i}
    {f(x)\rel{l_1}l_i}
  \end{mathpar}
  % Assume that the TSS has a support map given by
  % \begin{displaymath}
  %   \begin{array}{rcl}
  %     \eta(g(l_1))&=&\emptyset\\
  %     \eta(f(x))&=&\{g(x)\}.
  %   \end{array}
  % \end{displaymath}
  The TSS generates an LTS that is finitely branching since the provable
  transitions are just $g(l_1)\rel{l_1}l_1$ and $f(l_1)\rel{l_1}l_1$. Clearly,
  all the instantiations of the rule template on the right with $i\not=1$ are
  junk rules, because the axiom on the left only allows one to prove the
  transition $g(l_1)\rel{l_1}l_1$, which can only make true the premise of the
  instantiation with $i=1$. Pick some $\eta$. The $\eta$-type of rule
  \textsc{L} is
  \begin{displaymath}
    \langle g(l_1), \{t\mapsto\emptyset\mid t\in\eta(g(l_1))\}\rangle.
  \end{displaymath}
  The $\eta$-type of rule \textsc{2}$i$ is $\langle f(x), \psi_i\rangle$ where
  $\psi_i:\eta(f(x))\to\pset_\omega(L)$ is such that
  \begin{displaymath}
    \begin{array}{rcl}
      \psi_i(t)&=&
      \left\{
        \begin{array}{ll}
          \emptyset&\text{if}\ t\not=g(x)\\
          \{l_1\}&\text{if}\ t=g(x)
        \end{array}\right.
    \end{array}
  \end{displaymath}
  Since the premises of rules \textsc{2}$i$ with $i\in\NN$ only differ in
  their targets, all the $\psi_i$ are equal. Therefore, all the rules
  \textsc{2}$i$ with $i\in\NN$ have the same $\eta$-type $\langle f(x),
  \psi_i\rangle$, which is infinitely inhabited, and thus the TSS is not in
  the rule format of \cite{FV03}. In words, the rule format does not cover the
  TSS above, although the associated LTS is finitely branching. Notice that a
  rule format with the previously mentioned refinement of the restricted
  support in Observation~\ref{ob:restricted-support} would not cover the TSS
  either, because the support map $\eta$ defined as $\langle f(x),\{g(x)\mapsto\{l_1\}\}\rangle$
  meets the conditions of the
  restricted support (\ie, $\{g(x)\}\subseteq\eta(f(x))$), and every
  instantiation would have valid (refined) $\eta$-type.
\end{exmp}

\begin{obs}
  \label{ob:targets} {\it The $\eta$-types have to be refined in order to
    consider the targets of premises. Following this idea, for a rule with
    $\eta$-type $\langle t,\psi\rangle$, the map $\psi$ should deliver, for
    each source $u$ in the support $\eta(t)$, the set of pairs of labels and
    targets in the premises of the rule with source $u$.}
\end{obs}

Take the rules \textsc{2}$i$ with $i\in\NN$ of Example~\ref{ex:targets} and
consider the map $\eta$ given there. Each premise has a distinct target
$l_i$. Each rule \textsc{2}$i$ has $\eta$-type $\langle f(x), \psi_i\rangle$
where
\begin{displaymath}
  \begin{array}{rcl}
    \psi_i(t)&=&
    \left\{
      \begin{array}{ll}
        \emptyset&\text{if}\ t\not=g(x)\\
        \{(l_1,l_i)\}&\text{if}\ t=g(x)
      \end{array}\right.
  \end{array}
\end{displaymath}
and thus the $\eta$-types are finitely inhabited.

\subsection{Uniformity in the targets of premises}
\label{sec:uniform-targets}
Recall that the rule format in \cite{FV03} requires the TSS to be uniform.  We
paraphrase the notion of uniformity here as the requirement that if $t$ and
$t'$ are sources of any two rules in a TSS that differ only in the names of
the variables that occur in them, then $t$ and $t'$ are exactly the same term.
Uniformity prevents that several rules use distinct names in their sources for
variables that morally should have the same name. A non-uniform use of the
variables in the sources of rules may result in an infinitely branching LTS
even for a TSS whose $\eta$-types are all finitely inhabited. Since the
$\eta$-types of \cite{FV03} are only sensitive to the source of the rule
(notice that the labels are ground there), uniformity needed be enforced only
in the sources of rules. A non-uniform use of variables in other positions
would not jeopardize the finitely-inhabited discipline of the $\eta$-types of
\cite{FV03}, and the rule format there needed not care about those uses.

However, the refinement of the $\eta$-types in Observation~\ref{ob:targets}
that considers the targets of premises makes the $\psi$ map sensitive to the
targets of premises in rules. The following example shows that, additionally
to the variables in the source of the rules, the variables in the targets of
premises have to be used uniformly.
\begin{exmp}
  \label{ex:uniform-targets}
  Assume the existence of infinitely many different variable names $y_i$ with
  $i\in\NN$ that are distinct from variable name $x$. Consider the following
  TSS over the signature $\Sigma_0$ of Notation~\ref{not:sigma-zero}:
  \begin{mathpar}
    \inferrule*[right={\normalsize $,\quad\quad i\in\mathbb{N}$}]
    { }
    {l_i\rel{l_1}l_i}
    \and
    \inferrule*[right={\normalsize $,\quad\quad i\in\mathbb{N}$.}]
    {x\rel{l_1}y_i}
    {f(x)\rel{l_1}g^i(y_i)}
  \end{mathpar}
  Assume that the TSS has a support map given by
  \begin{displaymath}
    \begin{array}{rcll}
      \eta(l_i)&=&\emptyset\\
      \eta(f(x))&=&\{x\}\\
      \eta(t)&=&\emptyset&\text{\ow}
    \end{array}
  \end{displaymath}
  The premises of the different instantiations of the rule template on the
  right only differ in the variable name $y_i$ they use as a target of their
  premises. Each instantiation has a distinct target $g^i(y_i)$. For each $i$
  the instantiation of the rule template would have a different
  \mbox{$\eta$-type} $\langle f(x),{\{x\mapsto\{(l_1,y_i)\}\}}\rangle$
  (considering the refinement in Section~\ref{sec:targets} that takes the
  targets of premises into account) and the refined $\eta$-types would be
  finitely inhabited. In words, the non-uniform use of variable names $y_i$
  yields finitely inhabited $\eta$-types and the TSS would
  meet the conditions of the refined rule format. However, the associated LTS
  is not finitely branching because, for instance, $f(l_1)\rel{l_1}g^i(l_1)$ for
  each $i\in\NN$. Notice that this fact is also true if, additionally, we
  consider the refinement of the restricted support in
  Observation~\ref{ob:restricted-support}, because the set of sources of
  premises for each instantiation of the rule template on the right is a
  subset of the support for $f(x)$ (\ie, $\{x\}\subseteq\eta(f(x))$).

  % In order for the rule format to be
  % sound, different variable names in the targets should not be used to
  % distinguish different premisses, this is, the target should be used
  % uniformly. The TSS above can be written in a uniform way as follows:
  % \begin{mathpar}
  %   \inferrule*[right={\normalsize $,\quad\quad i\in\mathbb{N}$}]
  %   { }
  %   {l_i\rel{l_1}l_i}
  %   \and
  %   \inferrule*[right={\normalsize $,\quad\quad i\in\mathbb{N}$.}]
  %   {x\rel{l_1}y}
  %   {f(x)\rel{l_1}g^i(y)}
  % \end{mathpar}
  % Now the variable name `$y$' is the same in all the instantiations of the
  % rule template on the right, which have the same refined $\eta$-type
  % $\langle f(x),{\{x\mapsto\{y\}\}}\rangle$. This refined $\eta$-type is
  % infinitely inhabited and the TSS will be out of the conditions of the rule
  % format that we introduce in Section\ref{}.
\end{exmp}

\begin{obs}
  \label{ob:uniform-targets}
  {\it If the $\eta$-types are refined so that the map $\psi$ considers the
    sets of pairs of labels and targets in the premises whose sources are in
    the support, then the variables in the targets of premises have to be
    used uniformly.}
\end{obs}

If the TSS in Example~\ref{ex:uniform-targets} were uniform in the targets of
premises, then all the variables $y_i$ would be renamed to $y$. The TSS thus
obtained induces the same LTS as the one above. For every $\eta$, all the
instantiations of the rule template on the right have $\eta$-type $\langle
f(x),\psi\rangle$ where $\psi:\eta(f(x))\to\pset_\omega(L)$ is such that
\begin{displaymath}
  \begin{array}{rcl}
    \psi(t)&=&
    \left\{
      \begin{array}{ll}
        \emptyset&\text{if}\ t\not=x\\
        \{(l_1,y)\}&\text{if}\ t=x
      \end{array}\right.
  \end{array}
\end{displaymath}
and thus the $\eta$-types with the previously mentioned refinements would be
infinitely inhabited.

\section{Introducing the $S$-types}
\label{sec:s-types}
We introduce a variation of the $\eta$-types (we have called them $S$-types)
that incorporate the refinements in Observations~\ref{ob:restricted-support},
\ref{ob:targets} and \ref{ob:uniform-targets}, and that provide some of the
conditions of a rule format for finite branching that covers
Examples~\ref{ex:restricted-support}, \ref{ex:targets} and
\ref{ex:uniform-targets} among others. As far as we are aware, our
refinements are compatible with the TSSs that are covered already by the rule
format on \cite{FV03} (see Section~4 of \cite{FV03} for examples). We assume
triadic TSS with terms as labels and consider the dyadic transformations of
Section~\ref{sec:dyadic-transformation}. The conditions of the rule format are
given for the dyadic TSSs.

Recall the strict stratification of \cite{FV03}. A strict stratification is
defined for all closed terms. We notice that the strict stratification and the
information relative to the restricted support map of
Section~\ref{sec:restricted-support} can be neatly combined into a variation
of the strict stratification that, differently from the one in \cite{FV03},
allows certain terms that do not unify with sources of premises in rules to have
undefined order. We introduce the partial strict stratification of a dyadic
TSS.
\begin{ntn}
  Let $S$ be a partial map, we write $S(p)\not=\bot$ to indicate that $S(p)$
  is defined.
\end{ntn}

\begin{defn}[Partial strict stratification]
  \label{def:partial-strict-stratification}
  Let $R$ be a dyadic TSS and $S$ be a partial map from origins to ordinal
  numbers. $S$ is a partial strict stratification of $R$ iff the following
  conditions hold:
  \begin{enumerate}[(i)]
  \item $S(\sigma(s))\not=\bot$, for every source $s$ of some rule in $R$ and
    for every substitution $\sigma$ that closes $s$.
  \item For every rule $H/s\rel{}d$ in $R$ and for every $v\rel{}w\in H$, it
    holds that $S(\sigma(v))<S(\sigma(s))$ for each substitution $\sigma$ that
    closes $s$ and $v$ such that $S(\sigma(v))\not=\bot$.
  % \item For every rule in $R$ with source $s$ and set of premisses $H$, and
  %   for every $v\rel{}w\in H$, if $\sigma$ is a substitution that closes $s$
  %   and $v$ is such that $S(\sigma(v))\not=\bot$, then
  %   $S(\sigma(v))<S(\sigma(s))$.
  \end{enumerate}
  We say an origin $o$ has order $S(o)$.
\end{defn}

\begin{exmp}
  \label{ex:partial-strict-stratification}
  Consider the following dyadic TSS:
  \begin{mathpar}
    \inferrule*[left=L]
    { }
    {g(l_1)\rel{}(l_1,l_1)}
    \and
    \inferrule*[left=R$i$,right={\normalsize $,\qquad i\in\NN.$}]
    {g^i(x)\rel{}(l_i,x)}
    {f(x)\rel{}(l_1,x)}
  \end{mathpar}
  Define
  \begin{displaymath}
    \begin{array}{rcll}
      S(g(l_1))&=&0\\
      S(f(p))&=&1\\
      S(q)&=&\bot&\text{\ow}
    \end{array}
  \end{displaymath}
  % $S(g(l_1))=0$, $S(f(p))=1$ and $S(q)=\bot$ otherwise.
  Then $S$ is a partial strict stratification in the sense of
  Definition~\ref{def:partial-strict-stratification}.
\end{exmp}

Consider the rule \textsc{R}$2$ in the example above. The partial strict stratification
$S$ given in that example is not defined for terms of the form $g^2(p)$, which are the possible
instantiations of the source of the premise of that rule. This witnesses the fact that \textsc{R}$2$ 
is a junk rule. In general, a rule for which the sources of its premises do not unify with origins that
have a defined order is a junk rule. This is formalized by the following
lemma.
\begin{lem}[Junk rules]
  \label{lem:junk-rules}
  Let $R$ be a dyadic TSS and $S$ be a partial strict stratification of
  $R$. Let $o$ be an origin. Assume that $\rho=H/s\rel{}r\in R$ and
  $\sigma(s)=o$. If for every substitution $\tau$ such that $\tau\sigma$
  closes all the terms in $H$ there exists a premise $v\rel{}w \in H$ such
  that $S(\tau\sigma(v))=\bot$, then $\rho$ cannot unify with the root of a
  proof tree.
\end{lem}
\begin{proof}
  The term $\tau\sigma(v)$ does not unify with any rule $\rho'\in R$ because for
  all the origins $o'$ that unify with rules it must be the case that
  $S(o')\not=\bot$ by condition (i) in Definition  ~\ref{def:partial-strict-stratification}.
\end{proof}

The connection between a partial strict stratification and the restricted
support map is formalized as follows.
\begin{defn}[$S$-restricted support]
  \label{def:restricted-support}
  Let $R$ be a $\DT_k^\prj$-dyadic TSS and $S$ be a partial strict
  stratification of $R$. The $S$-restricted support map is the map
  $\eta:\oSource[\prj]{k}\to\pset(\oSource[\prj]{k})$ given by
  \begin{displaymath}
    \eta(s)=\{v\mid \exists\sigma.~S(\sigma(v))\not=\bot\
    \text{and $v\rel{}w$ is a premise in a rule with source
      $s$}\}.
  \end{displaymath}
  % \begin{align}
  %   \eta(s)=\{v\mid \exists\sigma.~ &S(\sigma(v))\not=\bot\
  %   \land\ S(\sigma(v))<S(\sigma(s))\nonumber\\
  %   &\land\ \text{$v\rel{}w$ is a premiss in some $\rho\in R$ with source
  %     $s$}\}.\nonumber
  % \end{align}
\end{defn}

For every $v$ in an $S$-restricted support $\eta(s)$, there is a premise of the form 
$v\rel{}w$ in some rule with source $s$ and $v$ unifies with an origin $o$ such that
$S(o)\not=\bot$. By Definition~\ref{def:partial-strict-stratification}, for
each $\sigma$ that closes both $v$ and $s$, $S(\sigma(v))<S(\sigma(s))$. 

For illustration, take the TSS in Example~\ref{ex:partial-strict-stratification}
and the partial strict stratification $S$ defined there. The $S$-restricted
support map is such that
\begin{displaymath}
  \begin{array}{rcll}
    \eta(g(l_1))&=&\emptyset,\\
    \eta(f(x))&=&\{g(x)\}&\text{and}\\
    \eta(t)&=&\emptyset&\text{\ow{.}}
  \end{array}
\end{displaymath}

Note that only rule \textsc{R}$1$ has a set of sources of premises that is included in $\eta(f(x))$.

The $S$-types that we introduce below rely on the notion of a partial strict
stratification. The $S$-types are no longer parametric on a support map $\eta$
(as was the case in \cite{FV03}), but on the partial strict stratification
$S$, which uniquely determines the $S$-restricted support map.
\begin{defn}[$S$-types]
  \label{def:S-types}
  Let $R$ be a $\DT_k^\prj$-dyadic TSS, $S$ be a partial strict stratification
  of $R$, $\eta:\oSource[\prj]{k}\to\pset(\oSource[\prj]{k})$ be the
  associated $S$-restricted support map, and $\rho=H/s\rel{}r$ be a rule of
  $R$ with source $s$ and premises ${H=\{v_i\rel{}w_i\mid i\in I\}}$.

  We say that $\rho$ has $S$-type $\langle s,\psi\rangle$ iff ${\{v_i\mid i\in
    I\}} \subseteq \eta(s)$, and the map $\psi(v)=\{w\mid v\rel{}w\ \text{is a
    premise of $\rho$}\}$ is such that
  $\psi:\eta(s)\to\pset_\omega(\oTarget[\prj]{k})$, that is, $\psi(v)$ is finite for
  each $v\in\eta(s)$.
\end{defn}
Notice that, for a rule $\rho=H/s\rel{}r$ to have a valid $S$-type $\langle
s,\psi\rangle$, the set of sources of premises in $H$ has to be a subset of
the support $\eta(s)$ and the elements in the codomain of $\psi$ have to be
finite sets.

Consider again the TSS in Example~\ref{ex:partial-strict-stratification},
the partial strict stratification $S$ defined there and the $S$-restricted
support $\eta$ given before. The $S$-type of the rule \textsc{L} is $\langle
g(l_1),\emptyset\rangle$, the $S$-type of the rule \textsc{R1} is $\langle
f(x),\{g(x)\mapsto\{(l_1,x)\}\}\rangle$, and the rules \textsc{R}$i$ where
$i\not=1$ do not have a valid $S$-type.

We consider uniformity both in the sources of rules and, as the refinement of
Section~\ref{sec:uniform-targets} requires, in the targets of premises. We
have adapted the notion of uniformity for TSSs from Definition~12 in
\cite{FV03} to dyadic TSSs. (In our previous work \cite{AcetoGI16}
  uniformity in the targets of premises is already considered when discussing
  a rule format for image finiteness in plain SOS.)
\begin{defn}[Uniform in the sources]
  \label{def:uniform-sources}
  Let $R$ be a dyadic TSS. $R$ is uniform in the sources iff for $s$ and $s'$
  sources of any two rules in $R$, either $s=s'$, or otherwise $s$ and $s'$
  cannot differ only in the names of their variables.
\end{defn}

\begin{defn}[Uniform in the targets of premises]
  \label{def:uniform-targets}
  Let $R$ be a dyadic TSS. $R$ is uniform in the targets of premises iff for
  $v\rel{}w$ and $v\rel{}w'$ premises of any two (not necessarily different)
  rules in $R$, either $w=w'$, or otherwise $w$ and $w'$ cannot differ only in
  the names of their variables.
\end{defn}

%%%%%%%%%%%%%%%%
In a TSS that is uniform in the sources (respectively, targets of premises),
each origin (respectively, destination) is a substitution instance of at most
finitely many sources (respectively, targets of positive premises) of
transition rules.
%%%%%%%%%%%%%%%%
\begin{prop}
  \label{pr:unify-finite-rules}
  Let $R$ be a $\DT_k^\prj$-dyadic TSS. The following statements hold:
  \begin{enumerate}[(i)]
  \item If $R$ is uniform in the sources, then for each closed term $p$ the
    set of pairs $(t,\sigma)$ with $\sigma:\var(t)\to\cTerm$ such that
    ${\sigma(t)=p}$ and $t$ is the source of some rule in $R$ is finite.
  \item If $R$ is uniform in the targets of positive premises, then for each
    transition $p\rel{}p'$, and for each term $t$ and substitution
    $\sigma:\var(t)\to\cTerm$ such that $\sigma(t)=p$, the set of pairs
    $(t',\tau)$ with $\tau:{(\var(t')\setminus\var(t))\to\cTerm}$ such that
    $\sigma(t)\rel{}\tau\sigma(t')=p\rel{}p'$ and $t\rel{}t'$ is a positive
    premise of some rule in $R$ is finite.
  \end{enumerate}
\end{prop}
We omit the proof of Proposition~\ref{pr:unify-finite-rules}, which can be
adapted straightforwardly from analogous results for triadic TSSs (see
Propositions~2 and 3 in our previous work presented in \cite{AcetoGI16}).

We follow \cite{FV03} and combine Definitions~\ref{def:S-types},
\ref{def:uniform-sources} and \ref{def:uniform-targets} into a single
condition that will be used later in the rule format of
Section~\ref{sec:rule-formats}.
\begin{defn}[$\DT_k^\prj$-bounded]\label{def:bounded-new}
  Let $R$ be a $\DT_k^\prj$-dyadic TSS. We say that $R$ is
  \mbox{$\DT_k^\prj$-bounded} iff $R$ is uniform both in the sources of rules
  and in the target of premises, and there exists a partial strict
  stratification $S$ of $R$ such that the elements in the codomain of the
  $S$-restricted support map are finite sets (\ie,
  $\eta:\oSource[\prj]{k}\to\pset_\omega(\oSource[\prj]{k})$) and for every
  rule $\rho\in R$ with $S$-type $\langle s,\psi\rangle$, the $S$-type
  $\langle s,\psi\rangle$ is finitely inhabited.

  We say that $R$ is $\DT_k^\prj$-bounded by $S$.
\end{defn}
The main result in Theorem~\ref{thm:dk-finiteness} to follow holds for every
partial strict stratification. However, different choices for the map $S$ may
filter out different sets of junk rules.

%%%%%%%%%%%%%%%%%%%%%%%%%%%%%%%%%%%%%%%%%%%%%%%%%%%%%%%%%%%%%%%%%%%%%%%%%%%
%% Rule format for Dk- and Dk-prj-finiteness
\section{Rule format for $\DT_k^\prj$-finiteness}
\label{sec:rule-formats}

Consider SOS rules of the form $c\rel{y}c$ or $c\rel{c}y$, where $c$ is a
constant and $y$ is a variable. Axioms of that form can be instantiated to
derive transitions $c\rel{p}c$ and $c\rel{c}p$ respectively, for each closed
term $p$, and constant $c$ would have infinitely many transitions. The bounded
nondeterminism format from \cite{FV03}, which enforces that all the variables
in the rules of a TSS are source dependent (see Definition~B.5.2 in
\cite{Fok00}), prevents the use of axioms of the form above. That is, a rule
in bounded nondeterminism format cannot introduce variables spuriously that
could break bounded nondeterminism. We adapt the bounded nondeterminism format
to dyadic TSSs and we rename it as $\DT_k^\prj$-bounded nondeterminism format.
\begin{defn}[$\DT_k^\prj$-bounded nondeterminism format]
  \label{def:bounded-nondeterminism-format}
  Let $R$ be a $\DT_k^\prj$-dyadic TSS. A rule in $R$ is in
  $\DT_k^\prj$-bounded nondeterminism format iff
  \begin{enumerate}[(i)]
  \item all the variables in the sources of premises occur also in the source
    of the rule, and
  \item all the variables in the target of the rule occur also in the source
    of the rule, or in the targets of its premises.
  \end{enumerate}
  A TSS is in $\DT_k^\prj$-bounded nondeterminism format if all its rules are
  in $\DT_k^\prj$-bounded nondeterminism format.
\end{defn}

Consider the above mentioned axioms $c\rel{y}c$ and $c\rel{c}y$. Their dyadic counterparts
generated using the transformation $\DT_1^{id}$ are $c\rel{}(y,c)$ and $c\rel{}(c,y)$, neither of
which is in bounded nondeterminism format. Applying the $\DT_4^{id}$ transformation to $c\rel{y}c$ and $c\rel{c}y$ yields the
dyadic rules $(c,y)\rel{}c$ and $(c,c)\rel{}y$, the former of which is in bounded nondeterminism format. Applying the
$\DT_1^{\pi_1}$ transformation to $c\rel{y}c$ and $c\rel{c}y$ yields the dyadic axioms $c\rel{}y$ and $c\rel{}c$, of which
the latter is in bounded nondeterminism format.

We are now ready to present our rule format for $\DT_k^\prj$-finiteness.
\begin{thm}
  \label{thm:dk-finiteness}
  Let $R$ be a $\DT_k^\prj$-bounded TSS in $\DT_k^\prj$-bounded nondeterminism
  format. The LTS associated with $R$ is finitely branching.
\end{thm}
\begin{proof}
  We prove that each origin $o$ in the LTS associated with $R$ is finitely
  branching. Since $R$ is $\DT_k^\prj$-bounded, it is uniform in the sources
  and there are only finitely many pairs $(s,\sigma)$ such that $\sigma(s)=o$
  and $s$ is the source of some rule of $R$.
  % (by \cite[Proposition~2]{AcetoGI16}).
  We focus on the set $\{(s_i,\sigma_i)\mid i\in I\}$, with $I$ a finite index
  set such that $o$ unifies with a rule that has some $s_i$ as source. Since
  $R$ is $\DT_k^\prj$-bounded, there exists a partial strict stratification
  $S$ such that $R$ is $\DT_k^\prj$-bounded by $S$. We proceed by induction on
  $S(o)$. (Recall from Definition~\ref{def:partial-strict-stratification} that
  $S(o)\not=\bot$.)

  The base case is when $S(o)=0$. The origin $o$ unifies with rules that have
  source $s_i$ and that may have valid $S$-type or not. By
  Lemma~\ref{lem:junk-rules}, the rules that do not have valid $S$-type cannot
  give rise to transitions and can be safely ignored. By
  Definition~\ref{def:partial-strict-stratification}, the rules that have
  valid $S$-type are of the form
  % then by Lemma~\ref{lem:junk-rules} the rules with source $s_i$ that give
  % rise to transitions are of the form
  \begin{mathpar}
    \inferrule*
    [right={\normalsize $,\quad\quad i\in I, ~j \in J_i$}]
    { }{s_i\rel{}r_j}
  \end{mathpar}
  where the $J_i$ are taken to be disjoint to avoid proliferation of
  indices. Since $R$ is $\DT_k^\prj$-bounded, for each $i\in I$ and each $j\in
  J_i$ the instantiation of the rule template above has $S$-type $\langle
  s_i,\psi_j\rangle$, and $\langle s_i,\psi_j\rangle$ is finitely inhabited.
  By Definitions~\ref{def:restricted-support} and \ref{def:S-types},
  $\psi_j=\{v\mapsto\emptyset\mid v\in\eta(s_i)\}$ for each $j\in J_i$. Since
  $R$ is in $\DT_k^\prj$-bounded nondeterminism format,
  $\var(r_j)\subseteq\var(s_i)$ and thus the $\sigma_i(r_j)$ are closed. Since
  for all $j\in J_i$ the $S$-types $\langle s_i,\psi_j\rangle$ are equal and
  they are finitely inhabited, the $J_i$ are finite. Therefore, for each $i\in
  I$ the set $\{\sigma_i(r_j)~|~{\sigma_i(s_i)\rel{}\sigma_i(r_j)}\ \text{with
    $j\in J_i$}\}$
  % \begin{displaymath}
  %   \{\sigma_i(r_j)~|~{\sigma_i(s_i)\rel{}\sigma_i(r_j)}\
  %   \text{with $j\in J_i$}\}
  % \end{displaymath}
  is finite. By the finiteness of $I$ it follows that the set
  $\{d~|~o\rel{}d\}$ is finite and the theorem holds.

  The general case is when $S(o)>0$. The rules with source $s_i$ are of the
  form
  \begin{equation}
    \label{ir:general}
    \makebox[0cm][c]{
      \begin{mathpar}
        \inferrule*
        [right={\normalsize $,\quad\quad i\in I,~j\in J_i$}]
        {\{v_\ell \rel{} w_\ell~|~\ell \in L_j\}}
        {s_i\rel{}r_j}
    \end{mathpar}}
  \end{equation}
  where the $J_i$ and the $L_j$ are taken to be disjoint to avoid
  proliferation of indices. (Note that $L_j$ may be empty for some rules.) It
  is safe to ignore all the rules that do not give rise to transitions and
  therefore, in the remainder of the proof, we assume that each rule $\rho$ of
  the form in (\ref{ir:general}) gives rise to transitions. We first show that
  $\rho$ has a valid $S$-type.

  Since $R$ is in $\DT_k^\prj$-bounded nondeterminism format,
  $\var(v_\ell)\subseteq\var(s_i)$ for each $\ell\in L_j$, and therefore the
  $\sigma_i(v_\ell)$ are closed terms. As $\rho$ gives rise to transitions, by
  Lemma~\ref{lem:junk-rules} we have that $S(\sigma_i(v_\ell))\not=\bot$. By
  Definitions~\ref{def:partial-strict-stratification} and
  \ref{def:restricted-support}, $\{v_\ell\mid \ell\in
  L_j\}\subseteq\eta(s_i)$. For every $v\in\eta(s_i)$, if $v$ is not a source
  $v_\ell$ of some premise, then by Definition~\ref{def:S-types}
  $\psi(v)=\emptyset$, which is a finite set. If $v$ is a source $v_\ell$ of
  some premise, then $S(\sigma_i(v))<S(\sigma_i(s_i))$ and by the induction
  hypothesis $\sigma_i(v)$ is finitely branching. Since $R$ is uniform in the targets of
  premises and by Proposition~\ref{pr:unify-finite-rules}, the set $\{w\mid
  v\rel{}w\ \text{is a premise of}\ \rho\}$ is a finite set. By
  Definition~\ref{def:S-types}, $\rho$ has a valid $S$-type. We let $\langle
  s_i,\psi_j\rangle$ be the $S$-type of $\rho$.

  Our goal is now to prove that there are only finitely many outgoing
  transitions that can be proved using rules of the form (1). To this end,
  first of all we show that there are only finitely many distinct $\psi_j$
  with $j\in J_i$ such that rules with $S$-type $\langle s_i,\psi_j\rangle$
  give rise to transitions from $\sigma_i(s_i)$. By
  Definition~\ref{def:S-types}, each rule of $S$-type $\langle
  s_i,\psi_j\rangle$ contains a premise of the form $v\rel{}w$ for each
  $v\in\eta(s_i)$ and each $w\in\psi_j(v)$. Since $R$ is in
  $\DT_k^\prj$-bounded nondeterminism format, $\var(v)\subseteq\var(s_i)$, and
  thus the $\sigma_i(v)$ are closed. By Definition~\ref{def:proof-tree}, for
  each transition in the node of a proof tree, if the transition unifies with
  a rule of $S$-type $\langle s_i,\psi_j\rangle$ then for each $v\in\eta(s_i)$
  the process $\sigma_i(v)$ can perform, at least, a transition for each
  $w\in\psi_j(v)$. The $\psi_j$ in the $S$-types of rules that give rise to
  transitions from $\sigma_i(s_i)$ are dependent functions of type
  $\Pi_{v\in\eta(s_i)}\{\sigma_i(v)\rel{}\tau\sigma_i(w)\}$ with substitutions
  $\tau:(\var(w)\setminus\var(v))\to\cTerm$.
  For each $i$ and each $j\in J_i$ the refined type of the $\psi_j$ is
  finitely inhabited, since the codomain of a dependent function depends on
  the inputs of the function. Each image of $\psi_j$ cannot be an arbitrary
  subset of $\oTarget[\prj]{k}$, but only the subset that is determined by the
  input $v$. That is, the only elements in the codomain of $\psi_j$ are the
  sets $\{\tau\sigma_i(w)\mid \sigma_i(v)\rel{}\tau\sigma_i(w)\}$ where
  $v\in\eta(s_i)$. Since the $\eta(s_i)$ are finite sets, both the domain and
  the codomain of $\psi_j$ are finite. Therefore, for each $i\in I$ there are
  only finitely many $\psi_j$ with $j\in J_i$ such that the rules with
  $S$-type $\langle s_i,\psi_j\rangle$ give rise to transitions from
  $\sigma_i(s_i)$.

  Finally, we consider a rule $\rho$ in the form of (\ref{ir:general}). Since
  $R$ is in $\DT_k^\prj$-bounded nondeterminism format, $\var(v_\ell)
  \subseteq \var(s_i)$ and therefore the $\sigma_i(v_\ell)$ are closed
  terms. Since $S(\sigma_i(v_\ell))<S(o)$, by the induction hypothesis, the
  $\sigma_i(v_\ell)$ are finitely branching, and therefore for each $i\in I$ the
  set $\{p'\mid \sigma_i(v_\ell)\rel{}p'\}$ is finite. Since $R$ is uniform in
  the targets of premises and by Proposition~\ref{pr:unify-finite-rules}, for
  each $i\in I$ the set $\{\tau_m\sigma_i(w_\ell)\mid
  \sigma_i(v_\ell)\rel{}\tau_m\sigma_i(w_\ell)\}$ is finite, with
  \begin{displaymath}
    \tau_m:((\bigcup_{\ell\in
      L_j}\var(w_\ell))\setminus\var(s_i))\to\cTerm
  \end{displaymath}
  closed substitutions where $m$ ranges over finite index sets $M_j$ with
  $j\in J_i$.
  Since $R$ is in $\DT_k^\prj$-bounded nondeterminism format,
  $\var(r_j)\subseteq(\var(s_i)\cup(\bigcup_{\ell\in L_j}\var(w_\ell)))$ and
  therefore the $\tau_m\sigma_i(r_j)$ are closed targets.
  Since for each $i$ there are finitely many distinct $\psi_j$ with $j\in J_i$
  such that the rules with $S$-type $\langle s_i, \psi_j\rangle$ give rise to
  transitions, and since the $\langle s_i, \psi_j\rangle$ are finitely
  inhabited and the $M_j$ are finite, then for each $i\in I$ the set
  $\{\tau_m\sigma_i(r_j)~|~{\sigma_i(s_i)\rel{}\tau_m\sigma_i(r_j)}\}$ is
  finite. By the finiteness of $I$ it follows that the set $\{d~|~o\rel{}d\}$
  is finite and we are done.
\end{proof}

\begin{rem}
  The use of dependent function types when showing that, for a given $i$,
  there are finitely many distinct $\psi_j$ with $j\in J_i$ such that the
  rules with $S$-type $\langle s_i,\psi_j\rangle$ give rise to transitions
  from $\sigma_i(s_i)$, follows from our previous work presented in
  \cite{AcetoGI16}. However, an alternative argument that does not involve
  dependent function types stems from \cite{FV03}. The alternative argument
  assumes that there exists some $m\in I$ for which there are infinitely many
  $\psi_n$ with $n\in J_m$ such that rules with $S$-type $\langle
  s_m,\psi_n\rangle$ give rise to transitions from $\sigma_m(t_m)$, and then
  shows that this assumption contradicts the induction hypothesis.
\end{rem}

\begin{cor}
  \label{cor:finite-branching}
  Let $R$ be a triadic TSS with terms as labels and let
  ${R'=\DT_k^\prj(R)}$. If $R'$ is $\DT_k^\prj$-bounded and in
  $\DT_k^\prj$-bounded nondeterminism format then $R$ is $\DT_k^\prj$-finite.
\end{cor}
\begin{proof}
  By Definition~\ref{def:DT-finite}.
\end{proof}

The conditions of Theorem~\ref{thm:dk-finiteness} define our rule format for
$\DT_k^\prj$-finiteness for triadic TSS with terms as labels.

%%%%%%%%%%%%%%%%%%%%%%%%%%%%%%%%%%%%%%%%%%%%%%%%%%%%%%%%%%%%%%%%%%%%%%%%%%%
%% Applications
\section{Application}
\label{sec:application}

\begin{figure}[t]
  \centering
  \begin{mathpar}
  \inferrule*
  { }
  {a\rel{z}_{/a}z}
  \and
  \inferrule*
  {a\not=b}
  {b\rel{z}_{/a}b}
  \and
  \inferrule*
  {x_0\rel{z}_{/a}y_0\\x_1\rel{z}_{/a}y_1}
  {c!x_0.x_1\rel{z}_{/a}c!y_0.y_1}
  \and
  \inferrule*
  {x\rel{z}_{/a}y\\a\not=b}
  {c?b.x\rel{z}_{/a}c?b.y}
  \\
  \inferrule*
  {x_0\rel{z}_{/a}y_0\\x_1\rel{z}_{/a}y_1}
  {x_0+x_1\rel{z}_{/a}y_0+y_1}
  \and
  \inferrule*
  {x_0\rel{z}_{/a}y_0\\x_1\rel{z}_{/a}y_1}
  {x_0\mid x_1\rel{z}_{/a}y_0\mid y_1}
  \\
  \inferrule*
  {x_0\rel{z}_{c?}y_0\\x_1\rel{z}_{c!}y_1}
  {x_0\mid x_1\rel{}_\tau y_0\mid y_1}
  \and
  \inferrule*
  {x_0\rel{z}_{c!}y_0\\x_1\rel{z}_{c?}y_1}
  {x_0\mid x_1\rel{}_\tau y_0\mid y_1}
  \\
  \inferrule*
  { }
  {c!x_0.x_1\rel{x_0}_{c!}x_1}
  \and
  \inferrule*
  {x_1\rel{z}_{/a}y_1}
  {c?a.x_1\rel{z}_{c?}y_1}
  \and
  \inferrule*
  { }
  {\tau.x\rel{}_\tau x}
  \\
  \inferrule*
  {x_0\rel{z}_{c!}y_0}
  {x_0+x_1\rel{z}_{c!}y_0}
  \and
  \inferrule*
  {x_0\rel{z}_{c?}y_0}
  {x_0+x_1\rel{z}_{c?}y_0}
  \and
  \inferrule*
  {x_0\rel{}_\tau y_0}
  {x_0+x_1\rel{}_\tau y_0}
  \\
  \inferrule*
  {x_1\rel{z}_{c!}y_1}
  {x_0+x_1\rel{z}_{c!}y_1}
  \and
  \inferrule*
  {x_1\rel{z}_{c?}y_1}
  {x_0+x_1\rel{z}_{c?}y_1}
  \and
  \inferrule*
  {x_1\rel{}_\tau y_1}
  {x_0+x_1\rel{}_\tau y_1}
  \\
  \inferrule*
  {x_0\rel{z}_{c!}y_0}
  {x_0\mid x_1\rel{z}_{c!}y_0\mid x_1}
  \and
  \inferrule*
  {x_0\rel{z}_{c?}y_0}
  {x_0\mid x_1\rel{z}_{c?}y_0\mid x_1}
  \and
  \inferrule*
  {x_0\rel{}_\tau y_0}
  {x_0\mid x_1\rel{}_\tau y_0\mid x_1}
  \\
  \inferrule*
  {x_1\rel{z}_{c!}y_1}
  {x_0\mid x_1\rel{z}_{c!}x_0\mid y_1}
  \and
  \inferrule*
  {x_1\rel{z}_{c?}y_1}
  {x_0\mid x_1\rel{z}_{c?}x_0\mid y_1}
  \and
  \inferrule*[right={\normalsize ,}]
  {x_1\rel{}_\tau y_1}
  {x_0\mid x_1\rel{}_\tau x_0\mid y_1}
  \\
  \text{where $a$ and $b$ are atoms and $c$ is a channel name.}
\end{mathpar}
\caption{Inference rules for MicroCHOCS.}  \label{fig:microCHOCS}
\end{figure}

The rule format of Section~\ref{sec:rule-formats} can be applied to TSSs with
terms as labels from the literature. As an example, we adapt the inference rules for CHOCS from Figure~1 in \cite{MGR05}. Consider
the subset of the CHOCS language in Figure~\ref{fig:microCHOCS}, which we call
MicroCHOCS, which includes operators for send $c!x.p$, receive $c?a.p$, choice
$p_1+p_2$ and parallel composition $p_1\mid p_2$. The following classes of
transitions are defined: $\rel{}_{/a}$ where $a$ is an atom, $\rel{}_{c!}$ and
$\rel{}_{c?}$ where $c$ is a channel name, and $\rel{}_\tau$ where $\tau$ is
the internal action. These transitions are labelled by arbitrary terms. Note
that the $x$, $z$, $x_0$, $x_1$, $y_0$ and $y_1$ are term variables. In words,
transition $p\rel{q}_{/a}p'$ means that process $p'$ is obtained by
substituting the free occurrences of atom $a$ in process
$p$ by process $q$. Transition $c!p.q\rel{p}_{c!}q$ means that process $c!p.q$ sends process
$p$ through channel $c$ and continues with process $q$. Transition
$c?a.p\rel{q}_{c?}p'$ means that process $c?a.p$ receives process $q$ through
channel $c$ and continues with process $p'$, which is the result of
substituting the free occurrences of atom $a$ in $p$ by $q$. Transition
$p\rel{}_{\tau}p'$ means that process $p$ continues with process $p'$ after
performing the silent action.

We will now use the theory we developed in the main body of the paper to study
the bounded nondeterminism properties of the transition relations in
MicroCHOCS.  We partition
the system into the different sub-systems defining each of the classes of
transitions. For each atom $a$, we will show that sub-system $\rel{}_{/a}$ is
image finite. For each channel name $c$, we will show that sub-system
$\rel{}_{c!}$ is finitely branching. For each channel name $c$, we will show
that sub-system $\rel{}_{c?}$ is image finite by considering a TSS with the
same cardinality that gets rid of transitions $\rel{}_{/a}$ and using the fact
that, for each atom $a$, the sub-system $\rel{}_{/a}$ is image finite. Finally,
we will show that subsystem $\rel{}_\tau$ is finitely branching by considering a
TSS with the same cardinality that gets rid of transitions $\rel{}_{c?}$ and
$\rel{}_{c!}$ and by using the fact that, for each channel name, the sub-system
$\rel{}_{c?}$ is image finite and the subsystem $\rel{}_{c!}$ is finitely
branching.

\begin{exmp}[Sub-system $\rel{}_{/a}$]
  \label{ex:subsitution}
  The first six rule templates of MicroCHOCS implement capture avoiding
  substitution, \ie, $t\rel{z}_{/a}t'$ where term $t'$ results from
  substituting subject $z$ for the free occurrences of atom $a$ in term
  $t$. Since the terms are finite, once the subject and the atom are fixed,
  there is only one resulting term and the LTS associated with these six rules
  is image finite. (Note, however, that the system is not finitely branching. 
  Indeed, $a\rel{p}_{/a}p$ for each closed term $p$.) This can be checked by applying the $\DT_4^\id$
  transformation and the rule format defined by the conditions of
  Theorem~\ref{thm:dk-finiteness}. We notice that the atom-inequality
  conditions in the second and in the fourth rules can only decrease the
  number of transitions in the associated LTS. Therefore it is safe to strip away
  the atom inequality conditions, which we omit from now on. We fix an atom
  $a$ and we let $R$ be the TSS consisting of the instances of the first six
  rules above that define the class of transitions $\rel{}_{/a}$. Below we
  give the $\DT_{4}^\id$-dyadic TSS $R'=\DT_4^\id(R)$:
  \begin{mathpar}
    \inferrule*
    { }
    {(a,z)\rel{}z}
    \and
    \inferrule*
    { }
    {(b,z)\rel{}b}
    \and
    \inferrule*
    {(x_0,z)\rel{}y_0\\(x_1,z)\rel{}y_1}
    {(c!x_0.x_1,z)\rel{}c!y_0.y_1}
    \\
    \inferrule*
    {(x,z)\rel{}y}
    {(c?b.x,z)\rel{}c?b.y}
    \and
    \inferrule*
    {(x_0,z)\rel{}y_0\\(x_1,z)\rel{}y_1}
    {(x_0+x_1,z)\rel{}y_0+y_1}
    \\
    \inferrule*[right={\normalsize .}]
    {(x_0,z)\rel{}y_0\\(x_1,z)\rel{}y_1}
    {(x_0\mid x_1,z)\rel{}y_0\mid y_1}
  \end{mathpar}
  $R'$ has a partial strict stratification $S$ given by
  \begin{displaymath}
    \begin{array}{rcll}
      S(b,q)&=&1&b\ \text{an atom}\\
      S(c!p_0.p_1,q)&=&1+S(p_0,q)+S(p_1,q)\\
      S(c?b.p,q)&=&1+S(p,q)\\
      S(p_0+p_1,q)&=&1+S(p_0,q)+S(p_1,q)\\
      S(p_0\mid p_1,q)&=&1+S(p_0,q)+S(p_1,q)\\
      S(p,q)&=&\bot&\text{\ow}
    \end{array}
  \end{displaymath}
  where $p$, $p_0$, $p_1$ and $q$ are closed terms. The $S$-restricted support
  map of $R'$ is given by
  \begin{displaymath}
    \begin{array}{rcll}
      \eta(b,z)&=&\emptyset &b\ \text{an atom}\\
      \eta(c!x_0.x_1,z)&=&\{(x_0,z),(x_1,z)\}\\
      \eta(c?b.x,z)&=&\{(x,z)\}\\
      \eta(x_0+x_1,z)&=&\{(x_0,z),(x_1,z)\}\\
      \eta(x_0\mid x_1,z)&=&\{(x_0,z),(x_1,z)\}.
    \end{array}
  \end{displaymath}
  The elements in the codomain of $\eta$ are finite sets. The rules above have
  respectively (in the order in which they are read) the $S$-types
  \begin{displaymath}
    \begin{array}{rl}
      \langle(a,z),&\emptyset\rangle\\
      \langle(b,z),&\emptyset\rangle\\
      \langle(c!x_0.x_1,z),&\{(x_0,z)\mapsto \{y_0\},
      (x_1,z)\mapsto \{y_1\}\}\rangle\\
      \langle(c?a.x,z),&\{(x,z)\mapsto \{y\}\}\rangle\\
      \langle(x_0+x_1,z),&\{(x_0,z)\mapsto \{y_0\},
      (x_1,z)\mapsto \{y_1\}\}\rangle\\
      \langle(x_0\mid x_1,z),&\{(x_0,z)\mapsto \{y_0\},
      (x_1,z)\mapsto \{y_1\}\}\rangle
    \end{array}
  \end{displaymath}
  which are finitely inhabited. Since $R'$ is in $\DT_4^\id$-bounded
  nondeterminism format, $R'$ meets the conditions of
  Theorem~\ref{thm:dk-finiteness}. Therefore $R'$ is finitely branching and by
  Corollary~\ref{cor:finite-branching} $R$ is \mbox{$\DT_4^\id$-finite}, that
  is, $R$ is image finite. Since the atom-inequality conditions that we
  removed from the second and fourth rules can only decrease the number of
  transitions, for each atom $a$, the LTS associated with the rules of
  MicroCHOCS that define the class of transitions $\rel{}_{/a}$ is image
  finite.
\end{exmp}

\begin{exmp}[Sub-system $\rel{}_{c!}$]
  \label{ex:send}
  The last five rule templates of the first column of the rules for MicroCHOCS in Figure~\ref{fig:microCHOCS} 
  describe the effect of sending a process on a channel. 
  They consist of an axiom and four compatibility rule
  templates for the choice and parallel composition operator. We check that
  this sub-system is finitely branching. We fix a channel $c$ and we let $R$ be
  the TSS consisting of the instances of the rules that define the class of
  transitions $\rel{}_{c!}$. Below we give the $\DT_{1}^\id$-dyadic TSS
  $R'=\DT_1^\id(R)$:
  \begin{mathpar}
    \inferrule*
    { }
    {c!x_0.x_1\rel{}(x_0,x_1)}
    \and
    \inferrule*
    {x_0\rel{}(z,y_0)}
    {x_0+x_1\rel{}(z,y_0)}
    \and
    \inferrule*
    {x_1\rel{}(z,y_1)}
    {x_0+x_1\rel{}(z,y_1)}
    \\
    \inferrule*
    {x_0\rel{}(z,y_0)}
    {x_0\mid x_1\rel{}(z,y_0\mid x_1)}
    \and
    \inferrule*[right={\normalsize .}]
    {x_1\rel{}(z,y_1)}
    {x_0\mid x_1\rel{}(z,x_0\mid y_1)}
  \end{mathpar}
  $R'$ has a partial strict stratification $S$ given by
  \begin{displaymath}
    \begin{array}{rcl}
      S(c!p_0.p_1)&=&1\\
      S(p_0+p_1)&=&1+S(p_0)+S(p_1)\\
      S(p_0\mid p_1)&=&1+S(p_0)+S(p_1)\\
      S(p)&=&\bot\qquad\text{\ow}
    \end{array}
  \end{displaymath}
  where $p$, $p_0$, $p_1$ and $q$ are closed terms. The $S$-restricted support
  map of $R'$ is given by
  \begin{displaymath}
    \begin{array}{rcl}
      \eta(c!x_0.x_1)&=&\emptyset\\
      \eta(x_0+x_1,z)&=&\{x_0,x_1\}\\
      \eta(x_0\mid x_1,z)&=&\{x_0,x_1\}.
    \end{array}
  \end{displaymath}
  The elements in the codomain of $\eta$ are finite sets. The rules above have
  respectively (in the order in which they are read) the $S$-types
  \begin{displaymath}
    \begin{array}{rl}
      \langle c!x_0.x_1,&\emptyset\rangle\\
      \langle x_0+x_1,&\{x_0\mapsto \{(z,y_0)\},x_1\mapsto\emptyset\}\rangle\\
      \langle x_0+x_1,&\{x_0\mapsto\emptyset,x_1\mapsto\{(z,y_1)\}\}\rangle\\
      \langle x_0\mid x_1,&\{x_0\mapsto \{(z,y_0)\},x_1\mapsto\emptyset\}\rangle\\
      \langle x_0\mid x_1,&\{x_0\mapsto\emptyset,x_1\mapsto\{(z,y_1)\}\}\rangle
    \end{array}
  \end{displaymath}
  which are finitely inhabited. Since $R'$ is in $\DT_1^\id$-bounded
  nondeterminism format, $R'$ meets the conditions of
  Theorem~\ref{thm:dk-finiteness}. Therefore $R'$ is finitely branching and, by
  Corollary~\ref{cor:finite-branching}, $R$ is \mbox{$\DT_1^\id$-finite}, that
  is, $R$ is finitely branching. For each channel name $c$, the LTS associated
  with the rules of MicroCHOCS that define the class of transitions
  $\rel{}_{c!}$ is finitely branching.
\end{exmp}

\begin{exmp}[Sub-system $\rel{}_{c?}$]
  \label{ex:receive}
  The last five rule templates of the second column of MicroCHOCS (together
  with the sub-system $\rel{}_{/a}$) describes the effect of receiving a process on a channel. They
  consist of a rule template that uses sub-system $\rel{}_{/a}$ in its premise
  and four compatibility rule templates for the choice and parallel
  composition operator. This sub-system is essentially the same as the one
  in Example~\ref{ex:send}, except for the first rule template
  \begin{mathpar}
    \inferrule*[right={\normalsize .}]
    {x_1\rel{z}_{/a}y_1}
    {c?a.x_1\rel{z}_{c?}y_1}
  \end{mathpar}
  Since for each atom $a$ the class of transitions $\rel{}_{/a}$ is image
  finite (Example~\ref{ex:subsitution}), in order to test the sub-system $\rel{}_{c?}$ for image finiteness
  it is enough to omit the rules for $\rel{}_{/a}$ and to replace the rule
  template above with the axiom template
  \begin{mathpar}
    \inferrule*[right={\normalsize .}]
    { }
    {c?a.x_1\rel{z}_{c?}x_1}
  \end{mathpar}
  We check that the LTS associated with the resulting TSS is image finite. We
  fix a channel $c$ and we let $R$ be the TSS that consists of the instances
  of the axiom above and of the compatibility rules that define the class of
  transitions $\rel{}_{c?}$. Let $R'=\DT_4^\id(R)$. Below we give the
  $\DT_{4}^\id$-dyadic TSS $R'=\DT_4^\id(R)$:
  \begin{mathpar}
    \inferrule*
    { }
    {(c?a.x_1,z)\rel{}x_1}
    \and
    \inferrule*
    {(x_0,z)\rel{}y_0}
    {(x_0+x_1,z)\rel{}y_0}
    \and
    \inferrule*
    {(x_1,z)\rel{}y_1}
    {(x_0+x_1,z)\rel{}y_1}
    \\
    \inferrule*
    {(x_0,z)\rel{}y_0}
    {(x_0\mid x_1,z)\rel{}y_0\mid x_1}
    \and
    \inferrule*[right={\normalsize .}]
    {(x_1,z)\rel{}y_1}
    {(x_0\mid x_1,z)\rel{}x_0\mid y_1}
  \end{mathpar}
  $R'$ has a partial strict stratification $S$ given by
  \begin{displaymath}
    \begin{array}{rcll}
      S(c?a.p,q)&=&1\\
      S(p_0+p_1,q)&=&1+S(p_0,q)+S(p_1,q)\\
      S(p_0\mid p_1,q)&=&1+S(p_0,q)+S(p_1,q)\\
      S(p,q)&=&\bot&\text{\ow}
    \end{array}
  \end{displaymath}
  where $p$, $p_0$, $p_1$ and $q$ are closed terms. $R'$ has an $S$-restricted
  support map given by
  \begin{displaymath}
    \begin{array}{rcl}
      \eta(c?a.x_1,z)&=&\emptyset\\
      \eta(x_0+x_1,z)&=&\{(x_0,z),(x_1,z)\}\\
      \eta(x_0\mid x_1,z)&=&\{(x_0,z),(x_1,z)\}.
    \end{array}
  \end{displaymath}
  The elements in the codomain of $\eta$ are finite sets. The rules above have
  respectively (in the order in which the rules are read) the $S$-types
  \begin{displaymath}
    \begin{array}{rl}
      \langle (c?a.x_1,z),&\emptyset\rangle\\
      \langle (x_0+x_1,z),&\{(x_0,z)\mapsto\{y_0\},
                            (x_1,z)\mapsto\emptyset\}\rangle\\
      \langle (x_0+x_1,z),&\{(x_0,z)\mapsto\emptyset,
                            (x_1,z)\mapsto\{y_1\}\}\rangle\\
      \langle (x_0\mid x_1,z),&\{(x_0,z)\mapsto\{y_0\},
                                 (x_1,z)\mapsto\emptyset\}\rangle\\
      \langle (x_0\mid x_1,z),&\{(x_0,z)\mapsto\emptyset,
                                 (x_1,z)\mapsto\{y_1\}\}\rangle
    \end{array}
  \end{displaymath}
  which are finitely inhabited. Since $R'$ is in $\DT_4^\id$-bounded
  nondeterminism format, $R'$ meets the conditions of
  Theorem~\ref{thm:dk-finiteness}. Therefore $R'$ is finitely branching and, by
  Corollary~\ref{cor:finite-branching}, $R$ is \mbox{$\DT_4^\id$-finite}, that
  is, $R$ is image finite. Since the LTS for $\rel{}_{/a}$
  is image finite for each atom $a$ and $R$ is image finite for each channel name $c$, then
  the LTS associated with the rules of MicroCHOCS that define the class of
  transitions $\rel{}_{c?}$ is image finite.
\end{exmp}

\begin{exmp}[Sub-system $\rel{}_{\tau}$]
  \label{ex:tau}
  The third row and the last five rules of the third column of MicroCHOCS
  (together with sub-systems $\rel{}_{c!}$ and $\rel{}_{c?}$) describe the 
  behaviour of the silent action $\tau$. This sub-system consist of the rules
  \begin{mathpar}
    \inferrule*
    {x_0\rel{z}_{c?}y_0\\x_1\rel{z}_{c!}y_1}
    {x_0\mid x_1\rel{}_\tau y_0\mid y_1}
    \and
    \inferrule*
    {x_0\rel{z}_{c!}y_0\\x_1\rel{z}_{c?}y_1}
    {x_0\mid x_1\rel{}_\tau y_0\mid y_1}
    \and
    \inferrule*
    { }
    {\tau.x\rel{}_\tau x}
  \end{mathpar}
  and four compatibility rules for the choice and parallel composition
  operator that resemble the ones in Examples~\ref{ex:send} and
  \ref{ex:receive} except that they omit the label $z$ that appears there.
  For each channel name $c$ the class of transitions $\rel{}_{c!}$ is finitely
  branching. Therefore, for the transitions of the form $p\mid
  q\rel{\tau}p'\mid q'$ that are provable by the rules in the third row of
  MicroCHOCS, there are only finitely many process $r$ that unify with the $z$
  that labels the premises of such rules. Once $r$ is fixed, since the class
  of transitions $\rel{}_{c?}$ is image finite, there are only finitely many
  processes $p'$ (notice that they unify with the $y_0$) and the process
  $p\mid q$, which unify with $x_0\mid y_0$, is finitely branching. In order to
  test the sub-system $\rel{}_{\tau}$ for finite branching it is therefore enough to
  consider the TSS $R$ that consists of the axioms
  \begin{mathpar}
    \inferrule*
    { }
    {x_0\mid x_1\rel{}_\tau x_0\mid x_1}
    \and
    \inferrule*[right={\normalsize .}]
    { }
    {\tau.x\rel{}_\tau x}
  \end{mathpar}
  and of the four compatibility rules. We check that the LTS associated with
  $R$ is finitely branching. Let $R'=\DT_1^\id(R)$ with the rules
  \begin{mathpar}
    \inferrule*
    { }
    {x_0\mid x_1\rel{} (0, x_0\mid x_1)}
    \and
    \inferrule*
    { }
    {\tau.x\rel{} (0,x)}
    \and
    \inferrule*
    {x_0\rel{}(0,y_0)}
    {x_0+x_1\rel{}(0,y_0)}
    \\
    \inferrule*
    {x_1\rel{}(0,y_1)}
    {x_0+x_1\rel{}(0,y_1)}
    \and
    \inferrule*
    {x_0\rel{}(0,y_0)}
    {x_0\mid x_1\rel{}(0,y_0\mid x_1)}
    \and
    \inferrule*
    {x_1\rel{}(0,y_1)}
    {x_0\mid x_1\rel{}(0,x_0\mid y_1)}
  \end{mathpar}
  where $0$ is the only constant process of MicroCHOCS. $R'$ has a partial
  strict stratification $S$ given by
  \begin{displaymath}
    \begin{array}{rcl}
      S(\tau.p)&=&1\\
      S(p_0+p_1)&=&1+S(p_0)+S(p_1)\\
      S(p_0\mid p_1)&=&1+S(p_0)+S(p_1)\\
      S(p)&=&\bot\qquad\text{\ow}
    \end{array}
  \end{displaymath}
  where $p$, $p_0$, $p_1$ and $q$ are closed terms. The $S$-restricted support
  map of $R'$ is given by
  \begin{displaymath}
    \begin{array}{rcl}
      \eta(\tau.x)&=&\emptyset\\
      \eta(x_0+x_1,z)&=&\{x_0,x_1\}\\
      \eta(x_0\mid x_1,z)&=&\{x_0,x_1\}.
    \end{array}
  \end{displaymath}
  The elements in the codomain of $\eta$ are finite sets. The rules above have
  respectively (in the order in which they are read) the $S$-types
  \begin{displaymath}
    \begin{array}{rl}
      \langle x_0\mid x_1,&\{x_0\mapsto\emptyset,x_1\mapsto\emptyset\}\rangle\\
      \langle \tau.x,&\emptyset\rangle\\
      \langle x_0+x_1,&\{x_0\mapsto \{(z,y_0)\},x_1\mapsto\emptyset\}\rangle\\
      \langle x_0+x_1,&\{x_0\mapsto\emptyset,x_1\mapsto\{(z,y_1)\}\}\rangle\\
      \langle x_0\mid x_1,&\{x_0\mapsto \{(z,y_0)\},x_1\mapsto\emptyset\}\rangle\\
      \langle x_0\mid x_1,&\{x_0\mapsto\emptyset,x_1\mapsto\{(z,y_1)\}\}\rangle
    \end{array}
  \end{displaymath}
  which are finitely inhabited. Since $R'$ is in $\DT_1^\id$-bounded
  nondeterminism format, $R'$ meets the conditions of
  Theorem~\ref{thm:dk-finiteness}. Therefore $R'$ is finitely branching and, by
  Corollary~\ref{cor:finite-branching}, $R$ is \mbox{$\DT_1^\id$-finite}, that
  is, $R$ is finitely branching. Since for each channel name $c$ the LTSs for
  $\rel{}_{c!}$ is finitely branching and $\rel{}_{c?}$ is image finite, the LTS
  associated with the rules of MicroCHOCS that define the class of transitions
  $\rel{}_\tau$ is finitely branching.
\end{exmp}

%%%%%%%%%%%%%%%%%%%%%%%%%%%%%%%%%%%%%%%%%%%%%%%%%%%%%%%%%%%%%%%%%%%%%%%%%%%
%% Future work
\section{Future work}
\label{sec:future-work}

The first condition in the $\DT_k^\prj$-bounded nondeterminism format from
Definition~\ref{def:bounded-nondeterminism-format} requires that all the
variables in the sources of premises occur also in the source of the
rule. This restriction disallows any sort of look-ahead in the rules, and it
is easy to construct TSSs that do not meet it and yield infinitely branching
LTSs. It is an interesting, and probably challenging, avenue for future work
to find ways of relaxing that constraint while preserving bounded
nondeterminism.

Another avenue for future work that we find worth pursuing is to investigate
whether the kind of dyadic TSSs we have considered in this paper may be used
fruitfully to study congruence rule formats for established notions of
bisimilarity for CHOCS. Earlier work on notions of bisimulation equivalence
based on presentations using \emph{commitments}~\cite{Mil93}, \emph{residual
  data types}~\cite{Ben10} and \emph{residuals}~\cite{PBEGW15} indicates that
such a programme might be successfully carried out using dyadic TSSs. However,
much work remains to be done in order to develop this line of work and to
assess fully the potential usefulness of dyadic formalisms in the study of
congruence properties of notions of bisimilarity for CHOCS-like languages.

In the remainder of this section we present TSSs that we are
aware are not covered by our rule format.

A partial strict stratification is not enough to deal with the general case of
a TSS with junk rules. The following example illustrates why.
\begin{exmp}
  Consider the TSS over the signature $\Sigma_0$ in
  Notation~\ref{not:sigma-zero}:
  \begin{mathpar}
    \inferrule*[right={\normalsize $,\qquad i\in\NN$}]
    {g^{i+1}(x)\rel{y}z}
    {g^i(x)\rel{y}z}
    \and
    \inferrule*[right={\normalsize .}]
    {g(x)\rel{y}z}
    {f(x)\rel{y}z}
  \end{mathpar}
  The TSS above does not contain any axiom, and therefore all the rules are
  junk rules (\ie, an infinite collection of premises with sources $g^i(x)$ and
  increasing $i$ could be stacked over a process $f(p)$ without ever
  constituting a proof tree). However, the TSS does not have a partial strict
  stratification because for each $i\in\NN$ the premise of the instantiation
  of the rule template on the left coincides with the source of the
  instantiation with $i+1$, and then the order of $g^i(p)$ has to be greater
  than the order of $g^{i+1}(p)$. This is not possible because we would need a
  partial strict stratification $S$ such that $S(g^1(p))>S(g^2(p))>\ldots$ 
  contradicting the well foundedness of the ordinals. In words, the TSS is finitely branching but
  does not meet the conditions of the rule format. For the rule format to
  cover this case, the partial strict stratification would have to allow
  processes unifying with the source of a rule to have undefined order, \eg,
  $S(f(p))=\bot$ and $S(g^i(p))=\bot$ with $i\in\NN$ and $p\in\cTerm$. But
  then proving that a junk rule cannot give rise to transitions poses a
  non-trivial challenge. 
%  We have devised a coinductive argument to do so, but
%  our proof is still unfinished and we leave this refinement as future work.
\end{exmp}

The bounded nondeterminism format (the $\DT_k^\prj$-bounded nondeterminism in
the dyadic case) is sometimes too restrictive, as shown by the following
example.
\begin{exmp}
  Consider the TSS over the signature $\Sigma_0$ in
  Notation~\ref{not:sigma-zero}:
  \begin{mathpar}
    \inferrule*[right={\normalsize $,\qquad i\in\NN$}]
    { }
    {f(x)\rel{l_1}l_1}
    \and
    \inferrule*[right={\normalsize .}]
    {f(x)\rel{y}z}
    {l_1\rel{y}z}
  \end{mathpar}
  For each process $f(p)$, the only provable transition is
  $f(p)\rel{l_1}l_1$. Therefore, the only transition from $l_1$ is $l_1\rel{l_1}l_1$. 
  However, the rule on the right is not in bounded
  nondeterminism format because its premise introduces variable $x$
  spuriously. In words, the TSS is finitely branching but does not meet the
  conditions of the rule format. The rule format should accept the rule on the
  right since the sub-process that unifies with $x$ (\ie, the $p$ in a process
  $f(p)$) is discarded by the next step in the transition sequence (\ie,
  $f(p)\rel{l_1}l_1$). We have not tackled this refinement yet.
\end{exmp}

Finally, Nominal Structural Operational semantics \cite{CMRG12} enriches the
SOS formalism by adopting the nominal techniques from \cite{GP99,UPG04} to
deal with names and many-sorted variable-binding operations within the SOS
framework. In a nominal TSS arbitrary terms can label transitions and the
rules may include a set of freshness assertions that restrict provability of
the rule. In order to apply our results to this setting we will need to extend
our rule format to many-sorted signatures and study the impact of the
freshness assertions. This is work in progress.

\paragraph{Acknowledgements} 
We thank the three anonymous referees for their careful reading of our
original submission and their insightful comments.

%%%%%%%%%%%%%%%%%%%%%%%%%%%%%%%%%%%%%%%%%%%%%%%%%%%%%%%%%%%%%%%%%%%%%%%%%%%
%% References
\section*{References}
\bibliographystyle{abbrv}
%\bibliography{rfbndnosos}
\providecommand{\noopsort}[1]{}

\end{document}